\documentclass[a4paper,12pt]{article}

\usepackage{diagrams}
\usepackage[top=1.0in,right=1in,left=1in,bottom=1.0in]{geometry}
\usepackage{enumerate}

\usepackage[dvips]{graphicx}
\usepackage[usenames]{color}
\usepackage{psfrag}

\usepackage{amsmath}
\usepackage{amssymb}
\usepackage{amsthm}

\theoremstyle{plain}
\newtheorem{theorem}{Theorem}[section]
\newtheorem{lemma}[theorem]{Lemma}

\theoremstyle{definition}
\newtheorem{defn}[theorem]{Definition}

\newcommand{\secretlabel}[1]{
\color{white}
\begin{equation}
\label{#1}
\end{equation}
\color{black}
\vspace{-41pt}
}

\newcommand{\Hom}{\mathrm{Hom}}
\newcommand{\Coh}{\mathrm{Coh}}

\renewcommand{\H}{\mathbf{H}}
\newcommand\swap{\mathrm{swap}}

\newcommand{\llangle}{\langle \hspace{-2pt} \langle}
\newcommand{\rrangle}{\rangle \hspace{-2pt} \rangle}

\newcommand\n{\ensuremath{\mathrm{n}}}

\newcommand\bigotimesn[1]{ \bigotimes \hspace{-2pt} \raisebox{-2pt}{${}_{#1}$} }

\newcommand{\cat}[1]{\ensuremath{\mathbf{#1}}}
\renewcommand{\a}[1]{a^{\phantom{\dag}}_{#1}}
\newcommand{\id}[1]{\mathrm{id}_{#1}}
\newcommand{\name}[1]{\ulcorner #1 \urcorner}

\newcommand\ts{\ensuremath{s}}

\newcommand{\pa}{\:\!}
\newcommand{\pb}{\;\!}
\newcommand{\ma}{\pb\!}
\newcommand{\mb}{\pa\!}

\renewcommand{\to}{\mbox{\begin{diagram}[width=4pt] {} & \rTo\end{diagram}}\ma}

\newarrow{mapsto}{|}{-}{-}{-}{>}

\newcommand{\centregraphics}[1]{\begin{array}{c}\includegraphics{#1}\end{array}}

\begin{document}

\title{A categorical framework for the \\quantum harmonic oscillator}
\author{Jamie Vicary
\\
Imperial College London
\\
\texttt{jamie.vicary@imperial.ac.uk}
}

\date{June 5, 2007}
\maketitle

\begin{abstract}

This paper describes how the structure of the state space of the quantum harmonic oscillator can be described by an adjunction of categories, that encodes the raising and lowering operators into a commutative comonoid. The formulation is an entirely general one in which Hilbert spaces play no special role.

Generalised coherent states arise through the hom-set isomorphisms defining the adjunction, and we prove that they are eigenstates of the lowering operators. Surprisingly, generalised exponentials also emerge naturally in this setting, and we demonstrate that coherent states are produced by the exponential of a raising morphism acting on the zero-particle state. Finally, we examine all of these constructions in a suitable category of Hilbert spaces, and find that they reproduce the conventional mathematical structures.
\end{abstract}

\section{Introduction}

With \cite{csqp}, a research programme was begun to describe many of the familiar aspects of quantum mechanics --- such as Hilbert spaces, unitary operators, states, inner products, superposition and entanglement --- in terms of the structure of a category. This programme suggests that we do quantum mechanics in an entirely new way: rather than explicitly working with the mathematics of Hilbert spaces, we should construct the category of Hilbert spaces and identify those parts of the categorical structure which are needed for the quantum mechanics that we want to do. Our quantum mechanics can then be done abstractly, in terms of this categorical structure. Such an approach leads to deeper insights into the mathematical structures which are necessary for quantum mechanics, as well as providing new avenues for its generalisation.

In this paper, we describe a way to extend the categorical description of quantum mechanics to encompass the quantum harmonic oscillator, one of the most basic and important quantum systems. We accomplish this with the theory of adjunctions, itself one of the most basic and important tools of category theory. The approach is based on the well-known fact, first discussed in \cite{llfock}, that the symmetric Fock space over a given single-particle Hilbert space is given by the canonical free commutative monoid object over that Hilbert space. In fact, we shall see that much more can be obtained from the free commutative monoid construction; every part of the adjunction structure corresponds naturally to some aspect of the traditional mathematical treatment of the quantum harmonic oscillator. A key observation is that the  natural isomorphisms
\begin{equation*}
F(A \oplus B) \simeq F(A) \otimes F(B)
\end{equation*}
should be unitary, where $F(A)$ represents the Fock space over a Hilbert space $A$, and $\oplus$ and $\otimes$ represent direct sum and tensor product respectively. The standard free commutative monoid functor will not, however, give rise to a unitary natural isomorphism in general.

We include introductions to the necessary category theory, and to the necessary physics, in the form of self-contained sections which can be read or skipped as desired.
\section{Categorical tools for quantum mechanics\label{tools}}

\subsection{Summary}

In this section we introduce the category theory that is useful for describing finite-dimensional quantum mechanics, which we will be using throughout this paper to work with the quantum harmonic oscillator. The structure that we will define by the end of this section is that of $\dag$-compact-closed categories with $\dag$-biproducts.

All of the structures in this section can be applied to the category $\cat{FdHilb}$, which has finite-dimensional Hilbert spaces as objects and bounded linear operators as morphisms. As we introduce each element of the structure, we will show how it can be formulated in this category, and describe how it captures some aspect of quantum mechanics as discussed in \cite{csqp}.

\subsection{Symmetric monoidal structure}
\label{symmon}
A symmetric monoidal category \cat{C} has a functor
\begin{equation*}
\otimes : \cat{C}\times \cat{C} \to \cat{C}
\end{equation*}
called the tensor product, and a monoidal unit object $I$, such that for all objects $A$, $B$ and $C$ in $\cat{C}$, the following isomorphisms exist:
\begin{align*}
\rho_A &: A \otimes I \simeq A
\\
\lambda_A &: I \otimes A \simeq A
\\
\alpha_{A,B,C} &: A \otimes (B \otimes C) \simeq (A \otimes B) \otimes C
\\
\mathrm{swap}^{\otimes}_{A,B} &: A \otimes B \simeq B \otimes A
\end{align*}
We also require that the $\swap ^\otimes$ isomorphisms satisfy the equation
\begin{equation*}
\swap ^{\otimes} _{B,A} \circ \swap ^\otimes _{A,B} = \id{A} \otimes \id{B}
\end{equation*}
for all objects $A$ and $B$. We require that all of these isomorphisms are natural; in other words, that they can be expressed as the stages of natural transformations. They must also be compatible with each other, so that any diagram built solely from these isomorphisms and their inverses is commutative. These isomorphisms imply that, up to isomorphism, $\otimes$ is a commutative, associative monoidal operation on our category, with unit object $I$. For readability we shall often exploit the coherence theorem \cite{cwm}, which proves that any monoidal category is equivalent to one for which all structural isomorphisms are identities; this allows us to neglect these isomorphisms (except for the $\swap ^\otimes$ isomorphism) when it is convenient. A good general reference for the theory of symmetric monoidal categories is \cite{cwm}.

The category $\cat{FdHilb}$ is a symmetric monoidal category, with the operation $\otimes$ given by the tensor product of Hilbert spaces. The tensor unit $I$ is the one-dimensional Hilbert space $\mathbb{C}$, the complex numbers. Observing that we can access the multiplicative monoid structure of the complex numbers by the endomorphisms of $\mathbb{C}$, we will refer to the endomorphism monoid $\Hom(I,I)$ in any monoidal category as the monoid of \emph{scalars}. It can be proved that the scalars in any monoidal category are commutative \cite{coherenceccc}.

Also, we note that the vectors in a finite-dimensional Hilbert space $A$ are in correspondence with bounded linear operators $\mathbb{C} \to A$ in $\cat{FdHilb}$. Since a state of a Hilbert space is given by a nonzero element of that Hilbert space, we can generalise the notion of state directly: given an arbitrary object $B$ in a symmetric monoidal category \cat{C}, the \emph{states} of $B$ are the nonzero morphisms $\Hom _\cat{C} (I, B)$. (We will encounter an abstract formulation of zero morphisms in section \ref{biproducts}.)

\subsection{Graphical representation}

Monoidal categories have a useful graphical representation. We represent objects of the category as  vertical lines, and the tensor product operation as lines placed side-by-side. Morphisms are junction boxes, with lines coming in and lines going out. For example, given objects $A$, $B$, and $C$, and morphisms $f: A\otimes B \to B \otimes A$ and $g: C \otimes A \to I$, we represent the composition
\begin{equation*}
(\id{B} \otimes g) \circ ( \id{B} \otimes \mathrm{swap} ^\otimes _{A,C}) \circ (f \otimes \id{C}) : A \otimes B \otimes C \to B
\end{equation*}
using the following diagram:
\begin{equation*}
{
\psfrag{f}{\hspace{-5pt}\raisebox{-1pt}{$f$}}
\psfrag{g}{\hspace{-5pt}\raisebox{0pt}{$g$}}
\psfrag{A}{$A$}
\psfrag{B}{$B$}
\psfrag{C}{$C$}
\includegraphics{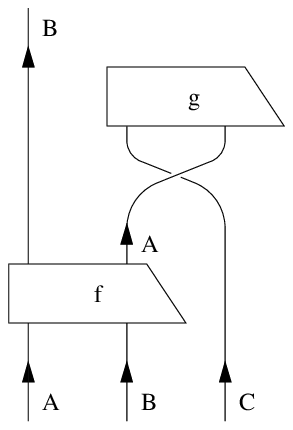}
}
\end{equation*}
The diagram is read from bottom to top. We represent the junction boxes using a wedge shape as we will later be rotating and reflecting them, and we want to break the symmetry so that their orientation can be identified. The tensor unit $I$ is represented by a blank space; in other words, it is not represented at all. The $\mathrm{swap}^{\otimes}$ isomorphism is represented by a crossing of lines, and because we are working in a symmetric monoidal category, it makes no difference which line goes over which. This graphical representation makes the structural isomorphisms $\rho$, $\lambda$, $\alpha$ and $\mathrm{swap} ^\otimes$ intuitive, and so is often a useful way to understand complex constructions in monoidal categories.

\subsection{Biproducts\label{biproducts}}

In some categories, products and coproducts become unified in a particularly elegant way.
\begin{defn} A category has a \emph{zero object}, written $0$, when it has isomorphic initial and terminal objects. Each hom-set $\Hom(A,B)$ then gains a \emph{zero morphism}, written $0_{A,B}$, which is the unique morphism in the hom-set that factors through the zero object.
\end{defn}
\begin{defn} A category has \emph{biproducts} if and only if, for all objects $A$ and $B$, the unique arrow $w_{A,B}:A + B \to A \times B$ that satisfies the equations
\begin{align*}
\pi_A \circ w_{A,B} \circ i_A &= \id{A}
\\
\pi_B \circ w_{A,B} \circ i_A &= 0_{A,B}
\\
\pi_A \circ w_{A,B} \circ i_B &= 0 _{B,A}
\\
\pi_B \circ w_{A,B} \circ i_B &= \id{B}
\end{align*}is an isomorphism, where $i_A$ and $i_B$ are the coproduct injections and $\pi_A$ and $\pi_B$ are the product projections.
\end{defn}

\noindent In a category with biproducts, we represent the isomorphic product and coproduct by the symbol $\oplus$, and call it a biproduct.

A category with biproducts is \emph{enriched over commutative monoids}; in other words, any hom-set $\Hom(A,B)$ carries the structure of a commutative monoid, with unit $0_{A,B}$. We interpret the monoid action as being addition, and we define it in the following way:
\begin{equation}
\begin{diagram}[height=25pt]
\label{biproductsum}
A&\rTo^{f+g}&B\\
\dTo<{\Delta_A}&&\uTo>{\nabla _B}\\
A\oplus A&\rTo_{f\oplus g}&B\oplus B
\end{diagram}
\end{equation}
Here, $\Delta_A$ is the diagonal for the product, and $\nabla_A$ is the codiagonal for the coproduct.

In a category with biproducts, we can always choose the canonical injections and projections to satisfy some useful properties, and when we talk about `the' injections and projections, it will be these well-behaved ones which are meant. For all objects \mbox{$B:=A_1\oplus A_2\oplus\ldots\oplus A_N$}, we can choose projection morphisms $\pi_n:B\to A_n$ and injection morphisms $i_n:A_n\to B$, such that the following properties hold:
\begin{enumerate}
\item $\pi_{n} \circ i_{n}=\id{A_n}$ for all $n$;
\item $\pi_{m}\circ i_{n}=0_{A_n,A_m}$ when $n\neq m$;
\item $\sum_{n=1}^N(i_{n}\circ \pi_{n})=\id{B}$, where this sum is as defined in diagram (\ref{biproductsum}).
\end{enumerate} 
In fact, in a category which is enriched over commutative monoids, this can be taken as the definition of a biproduct structure.

Naturality of the diagonal and codiagonal morphisms implies that in a category with biproducts, composition is linear. In other words, for all $f,f':A \to B$ and \mbox{$g,g':B \to C$}, we have
\begin{equation*}
g \circ (f + f') = (g \circ f) + (g \circ f')
\qquad\qquad
(g+g') \circ f = (g \circ f) + (g ' \circ f)
\end{equation*}

The scalars in the category interact nicely with a biproduct structure. Not only can we add elements of hom-sets, but we can also multiply them by scalars in a well-defined way.
\begin{defn} For any morphism $f:A \to B$, for any objects $A$ and $B$, and for any scalar $s: I \to I$, the scalar multiple $s \cdot f$ is defined in the following way\footnote{Of course, there are other equivalent definitions.}:
\begin{equation}
\begin{diagram}[midshaft,height=25pt]
A & \rTo ^{s \cdot f} & B
\\
\dTo <{\lambda _A ^{-1}} & & \uTo>{\lambda _B}
\\
I \otimes A & \rTo_{s \otimes f} & I \otimes B
\end{diagram}
\end{equation}
\end{defn}
\noindent In fact, the scalars interact with the biproduct structure in such a way that each hom-set gains the structure of a \emph{commutative semimodule}, a weakening of the notion of a vector space. In particular, for all scalars $s$ and $p$ and all morphisms $f$, we have
\begin{equation*}
(s \circ p) \cdot f = s \cdot (p \cdot f).
\end{equation*}
Scalar multiplication is also well-behaved with respect to composition, meaning that for any $g$ with $g \circ f$ well-defined, we have
\begin{equation*}
g \circ (s \cdot f) = (s \cdot g) \circ f = s \cdot (g \circ f).
\end{equation*}

The category $\cat{FdHilb}$ has biproducts, given by the direct sum of Hilbert spaces. In this case, the $i_n$ are injections of subspaces, and the $\pi_n$ are projections of subspaces. In an arbitrary symmetric monoidal category with biproducts, the commutative monoid structure induced on the hom-sets gives us a way to add states: given arbitrary \mbox{$\phi,\psi:I \to A$}, the superposition is given by \mbox{$\phi+\psi:I \to A$}. In $\cat{FdHilb}$, this agrees with the usual notion of addition of vectors.

\subsection{$\dag$-categories}
\label{daggerstructure}
\begin{defn} A category \cat{C} is a \emph{$\dag$-category} if it is endowed with a contravariant endofunctor $\dag:\cat{C} \to \cat{C}$, which is the identity on objects, and which satisfies $\dag \circ \dag = \id{\cat{C}}$.
\end{defn}
\noindent A category might admit multiple endofunctors satisfying this condition, but if we refer to it as a $\dag$-category then we must have a particular one in mind, which we will refer to as $\dag$. We also note that any $\dag$-category must be isomorphic to its opposite.

We can use this structure to adapt some useful terminology from the mathematics of Hilbert spaces.
\begin{defn} For any morphism $f$, we call $f ^\dag$ its \emph{adjoint}.
\end{defn}
\begin{defn} A morphism $f:A \to A$ is \emph{self-adjoint} if it satisfies $f ^\dag = f$. 
\end{defn}
\begin{defn} A morphism $f:A \to B$ is an \emph{isometry} if it satisfies $f^\dag \circ f = \id{A}$; in other words, if its adjoint is its retraction.
\end{defn}
\begin{defn} A morphism $f:A \to B$ is \emph{unitary} if it satisfies $f^\dag \circ f = \id{A}$ and $f \circ f^\dag = \id{B}$; in other words, if both $f$ and $f ^\dag$ are isometries. In this case, $A$ and $B$ are of course isomorphic.
\end{defn}

If a $\dag$-category has additional structure, we will often require that the additional structure be compatible with the $\dag$ functor.

\begin{defn} A $\dag$-category has \emph{$\dag$-biproducts} iff it has biproducts, such that the canonical projections and injections are related by the $\dag$ functor.
\end{defn}
\begin{defn} A \emph{symmetric monoidal $\dag$-category} is defined in the obvious way, but with the extra constraints that the canonical isomorphisms associated to the symmetric monoidal structure be unitary.
\end{defn}
\begin{defn}[Notation due to Selinger \cite{selingeridempotents}]
An equaliser $d:D \to A$ is a \mbox{\emph{$\dag$-equaliser}} if $d$ is an isometry (note that it must automatically be monic, by the properties of the equaliser.)
\end{defn}

Finally, we note that natural transformations defined between functors connecting $\dag$-categories may themselves admit a notion of adjoint.
\begin{defn} For $\dag$-categories \cat{C} and \cat{D}, a functor $J:\cat{C} \to \cat{D}$ is \emph{compatible} with the $\dag$-structures iff for all morphisms $f$ in \cat{C}, we have $J(f^\dag) = (J(f)) ^\dag$, where we employ the $\dag$ on \cat{C}\ and \cat{D} respectively.
\end{defn}
\begin{lemma}
Given a natural transformation $n: J \dot{\to} K$ for functors $J,\,K:\cat{C} \to \cat{D}$ compatible with $\dag$-structures on \cat{C} and \cat{D}, then $n$ has an adjoint natural transformation $n ^\dag : K \dot{\to} J$, defined at each stage $A$ of $\cat{C}$ by $(n^\dag) _A := (n_A)^\dag$.
\end{lemma}

\subsection{Duals and $\dag$-compact closure}
\label{dualssection}

A symmetric monoidal category has \emph{duals}, or equivalently is \emph{compact-closed}, if for every object $A$ there exists a second object $A^*$ and morphisms
\begin{equation*}
\zeta _A:I\to A\otimes A^* \quad\quad\quad\quad \theta _A:A^*\otimes A \to I
\end{equation*}
which satisfy the following diagrams:
\begin{gather}
\label{dualeq1}
\begin{diagram}[midshaft,width=59pt,height=15pt]
A & \rTo^{\zeta_A \otimes \id{A}} & A \otimes A^* \otimes A & \rTo^{\id{A} \otimes \theta _A} & A
\\
\dEq & & & & \dEq
\\
A & \rTo^{\id{A}} & & & A
\end{diagram}
\\\nonumber
\\
\label{dualeq2}
\begin{diagram}[midshaft,width=55pt,height=15pt]
A^* & \rTo^{\id{A^*} \otimes \zeta _A} & A^* \otimes A \otimes A^* & \rTo^{\theta _A \otimes \id{A^*}} & A^*
\\
\dEq & & & & \dEq
\\
A^* & \rTo^{\id{A^*}} & & & A^*
\end{diagram}
\end{gather}
These equations imply that $A^*$ is unique up to isomorphism, and that $(A^*)^* \simeq A$. However, it is useful to talk about $A^*$ as if it were unique, and to use $(A^*)^*=A$ as if it held as an equation, knowing that what we do will only hold up to isomorphism.
We refer to $A^*$ as \emph{the dual} of $A$.

We can use the duals to define a natural isomorphism of hom-sets
\begin{equation}\label{homsetiso}
S_{A, B, C} : \Hom(A\otimes B,C) \simeq \Hom(B,C\otimes A^*)
\end{equation}
for all objects $A$, $B$ and $C$, as shown in the following diagrams for any $f:A\otimes B \to C$ and $f':B \to C \otimes A^*$ related by the isomorphism:
\begin{gather}
\begin{diagram}[midshaft,width=90pt,height=25pt]
B & \rTo^{f' = S_{A, B, C}(f)} & C \otimes A^*
\\
\dTo < {\id{B} \otimes \zeta _A} & & \uTo>{f \otimes \id{A^*}}
\\
B \otimes A \otimes A^* & \rTo_{\mathrm{swap} ^\otimes _{B,A} \otimes \id{A^*}} & A \otimes B \otimes A^*
\end{diagram}
\\\nonumber
\\
\begin{diagram}[midshaft,width=90pt,height=25pt]
A \otimes B & \rTo^{f = S ^{-1} _{A, B, C}(f')} & C
\\
\dTo < {\mathrm{swap} ^\otimes _{A,B}} & & \uTo >{\id{C} \otimes \theta _A}
\\
B \otimes A & \rTo_{f' \otimes \id{A}} & C \otimes A^* \otimes A
\end{diagram}
\end{gather}
Here, the morphisms of the form $\mathrm{swap} ^\otimes  _{A,B}$ are the symmetry isomorphisms that make up part of the symmetric monoidal structure. Equations (\ref{dualeq1}) and (\ref{dualeq2}) ensure that the diagrams are inverse to each other in the correct way.

The graphical representation of the dual structure is especially powerful. If an object $A$ is represented by a line with arrow oriented up the page, then its dual $A^*$ has an arrow pointing down the page, and vice-versa. The duality morphisms $\zeta _A$ and $\theta _A$ then take the form of lines which loop back on themselves:
\begin{equation*}
\begin{array}{cc}
{
\psfrag{A}{$A$}
\psfrag{As}{$A^*$}
\includegraphics{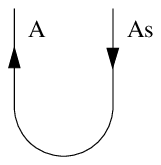}
}
&
{
\psfrag{A}{$A$}
\psfrag{As}{$A^*$}
\includegraphics{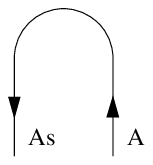}
}
\\
\qquad
\zeta _A: I \to A \otimes A^*
\qquad
&
\qquad
\theta _A:A^* \otimes A \to I
\qquad
\end{array}
\end{equation*}
Equations (\ref{dualeq1}) and (\ref{dualeq2}) then simply state that curves in the graphical representation can be straightened out:
\begin{equation*}
\psfrag{A}{$A$}
\psfrag{As}{$A^*$}
\begin{array}{c}
\includegraphics{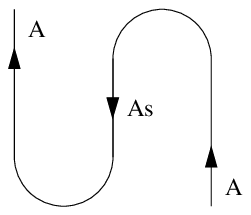}
\hspace{-9pt}
\end{array}
=
\begin{array}{c}
\includegraphics{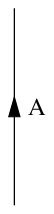}
\end{array}
\qquad\qquad
\begin{array}{c}
\includegraphics{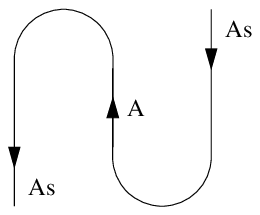}
\hspace{-17pt}
\end{array}
=
\begin{array}{c}
\includegraphics{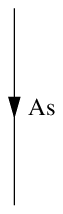}
\end{array}
\end{equation*}
This is very intuitive; when working in monoidal categories with duals, it is often much easier to work with expressions graphically rather than symbolically, as the eye can easily spot simplifications.

The isomorphism (\ref{homsetiso}) also has a straightforward interpretation in the graphical representation. To obtain $S_{A,B,C}(f): B \to C \otimes A^*$ from $f:A \otimes B \to C$, one simply `bends around' the line representing the object $A$:
\begin{equation*}
\psfrag{A}{$A$}
\psfrag{As}{$A^*$}
\psfrag{B}{$B$}
\psfrag{C}{$C$}
\psfrag{f}{\raisebox{-1pt}{\hspace{-8pt}$f$}}
\begin{array}{ccc}
\includegraphics{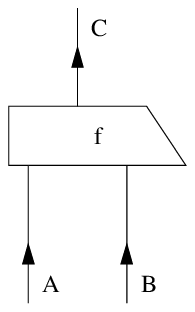}
&\qquad&
\includegraphics{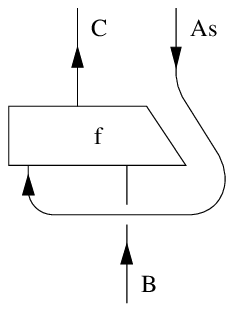}
\hspace{-10pt}
\\
\quad
f: A \otimes B \to C
\quad
&\qquad&
\quad
\theta_{A,B,C}(f): B \to C \otimes A^*
\quad
\end{array}
\end{equation*}
The inverse to the isomorphism bends the line around again, and by the `straightening-out' rule this gives back the original morphism $f$.

There are many things we can achieve by bending lines. For example, for any morphism \mbox{$h:A \to B$}, we could bend around both the $A$ line and the $B$ line. This gives a morphism \mbox{$h^*: B^* \to A^*$}, which we call the \emph{transpose} or \emph{dual} of $h$. This operation is involutive, so $(h^*)^* = h$. It is also functorial, in the sense that it defines a contravariant functor
\[
(-)^*:\cat{C^\mathrm{op}} \to \cat{C}
\]
satisfying $((-)^*)^* = \id{\cat{C}}$. We call this the \emph{duality functor}. In our graphical representation, there is a simple way to represent the duality:
\begin{equation*}
\psfrag{As}{$A^*$}
\psfrag{Bs}{$B^*$}
\psfrag{h}{\hspace{-5pt}\raisebox{-2pt}{$h$}}
\psfrag{hs}{\hspace{6pt}\raisebox{-2pt}{$h^*$}}
\begin{array}{c}
\includegraphics{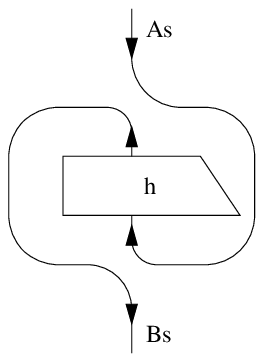}
\end{array}
=
\begin{array}{c}
\includegraphics{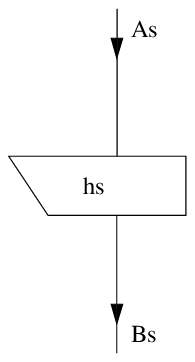}
\end{array}
\end{equation*}
We `straighten out the lines', rotating the junction box for $h$ by $180 ^\circ$ as we do so.

Given $h:A \to B$, we could also choose to bend around just the $A$ line, giving a morphism $\name{h}:I \to B \otimes A^*$. We call this the \emph{name} of $h$. Given a second morphism $k:B \to C$, we can perform the composition $k \circ h$ in terms of the names $\name{h}$ and $\name{k}$ by applying the equation $(\id{C} \otimes \theta _B \otimes \id{A^*}) \circ (\name{k} \otimes \name{h})=\name{k \circ h}$, which is simple to prove using the graphical representation:
\begin{equation*}
\psfrag{k}{\hspace{-5pt}\raisebox{-2pt}{$k$}}
\psfrag{h}{\hspace{-5pt}\raisebox{-2pt}{$h$}}
\psfrag{kch}{\hspace{-10pt}\raisebox{-1.5pt}{$k \circ h$}}
\psfrag{A}{$A$}
\psfrag{As}{$A^*$}
\psfrag{C}{$C$}
\psfrag{B}{$B$}
\psfrag{Bs}{$B^*$}
\begin{array}{c}
\begin{diagram}[height=25pt,labelstyle=\scriptstyle]
C \otimes A^*
\\
\uTo<{\id{C} \otimes \theta_B \otimes \id{A^*}}
\\
C \otimes B^* \otimes B \otimes A^*
\\
\uTo<{\name{k} \otimes \name{h}}
\\
I
\end{diagram}
\end{array}
\hspace{25pt}
=
\begin{array}{c}
\hspace{-5pt}
\includegraphics{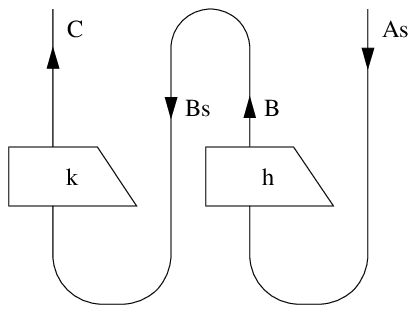}
\hspace{-24pt}
\end{array}
=
\begin{array}{c}
\hspace{-5pt}
\includegraphics{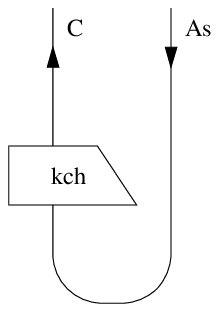}
\hspace{-24pt}
\end{array}
=
\begin{array}{c}
\begin{diagram}[height=25pt,labelstyle=\scriptstyle]
C \otimes A^*
\\
\uTo<{\name{k \circ h}}
\\
I
\end{diagram}
\end{array}
\end{equation*}

A compact-closed category is \emph{$\dag$-compact-closed}, or is a \emph{$\dag$-category with duals}, if there also exists  a \emph{conjugation functor}
\begin{equation}
(-)_* : \cat{C} \to \cat{C}
\end{equation}
satisfying $((-)_*)_* = \id{\cat{C}}$, which is compatible with the duality functor, in the sense that $((-)_*)^* = ((-)^*)_*$. This composite is then called the \emph{adjoint functor} or \emph{$\dag$-functor}, as described in section \ref{daggerstructure}, and is denoted in the following way:
\begin{equation}
(-)^\dag := ((-)^*)_*.
\end{equation}
We extend our graphical notation in a standard way to represent $(-)_*$ as flipping over a vertical axis, and changing the direction of arrows. A vertical-axis flip commutes with a $180 ^\circ$ rotation to produce a horizontal-axis flip, and so we use a horizontal-axis flip along with a change of arrow direction to represent $(-)^\dag$.

All of these structures appear in $\cat{FdHilb}$. The dual of a Hilbert space $A$ given by is its dual in the usual sense, the Hilbert space of bounded linear operators $A \to \mathbb{C}$. Given any basis $v_n$ of $A$, we obtain a basis of $A^*$ from the linear operators $v_n^* := (-,v_n) _A$, where $({-},{-}) _A$ is the inner product on $A$. We can then define the unit and counit as $\zeta _A =\sum_n v_n \otimes v_n^*$ and $\theta _A \circ ( v_n^* \otimes v_m) = \delta_{nm}$ respectively. These are bounded linear operators, and so are valid morphisms in the category\footnote{Unfortunately, they are only bounded because the underlying Hilbert spaces are finite-dimensional. The infinite-dimensional case raises significant difficulties, which we discuss in section \ref{conventional}.}. The transpose of a morphism is the matrix transpose in the usual sense, and the conjugation functor $(-)_*$ is complex conjugation. This then produces the adjoint functor $(-)^\dag$ as the familiar matrix conjugate-transpose operation.

Given the existence of an isomorphism $A \simeq A^*$, we can interpret the unit \mbox{$\zeta_A  : I \to A \otimes A^*$} as preparation of a particular Bell entangled state, and the counit $\theta _A : A^* \otimes A \to I$ as performing a Bell measurement. As discussed in \cite{csqp}, this allows use of the graphical calculus to aid design of quantum algorithms, such as the quantum teleportation protocol.

\subsection{Internal monoids and comonoids}

A conventional monoid is built from a set $S$ of elements, along with a multiplication map \mbox{$g:S \times S \to S$}, where $\times$ is the cartesian product, and a unit map $u:1 \to S$, where 1 is the one-element set. The associativity, unit, and commutativity laws can all then be phrased as categorical diagrams which $g$ and $u$ must satisfy.

However, this can be generalised: given any monoidal category \cat{C}, we can replace the cartesian product $\times$ in this definition with the monoidal product $\otimes$, and the one-element set $1$ with the monoidal unit object $I$. In a symmetric monoidal category, an \emph{internal commutative monoid} $(A,g,u)_+$ consists of an object $A$, a multiplication morphism \mbox{$g:A \otimes A \to A$} and a unit morphism $u:I \to A$, which satisfy associativity, unit and commutativity diagrams. If the context is clear, we shall often just refer to them as monoids. We will use the following graphical representation for the multiplication and unit morphisms for a monoid:
\begin{equation*}
\begin{array}{c}
\includegraphics{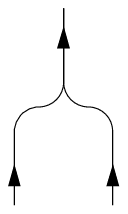}
\\
g:A \otimes A \to A
\end{array}
\qquad\qquad
\begin{array}{c}
\includegraphics{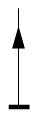}
\\
u: I \to A
\end{array}
\end{equation*}
Here, the vertical lines are all instances of the object $A$.  In terms of this graphical representation, the associativity, unit and commutativity laws are as follows:
\secretlabel{assoc}
\secretlabel{unit}
\secretlabel{comm}
\begin{gather*}
\begin{array}{c}
\begin{array}{c}
\includegraphics{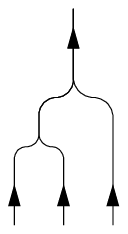}
\end{array}
\hspace{-5pt}
=
\hspace{-5pt}
\begin{array}{c}
\includegraphics{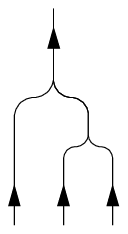}
\end{array}
\\
\textrm{Associativity law (\ref{assoc})}
\end{array}
\qquad
\begin{array}{c}
\begin{array}{c}
\includegraphics{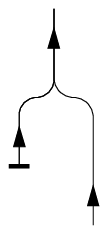}
\end{array}
\hspace{-5pt}
=
\hspace{-5pt}
\begin{array}{c}
\includegraphics{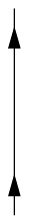}
\end{array}
\hspace{-5pt}
=
\hspace{-5pt}
\begin{array}{c}
\includegraphics{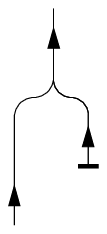}
\end{array}
\\
\textrm{Unit laws (\ref{unit})}
\end{array}
\qquad
\begin{array}{c}
\begin{array}{c}
\includegraphics{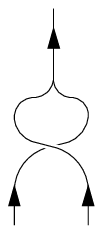}
\end{array}
\hspace{-5pt}
=
\hspace{-5pt}
\begin{array}{c}
\includegraphics{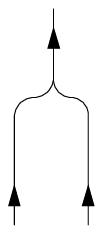}
\end{array}
\\
\textrm{Commutativity law (\ref{comm})}
\end{array}
\end{gather*}
The dual notion is an \emph{internal comonoid}. An internal comonoid $(A,h,v) _\times$ consists of an object $A$, a comultiplication morphism $h:A \to A \otimes A$ and a counit morphism $v:A \to I$, such that these morphisms satisfy the coassociativity, counit and cocommutativity laws, which are just the associativity, unit and commutativity laws with the arrows reversed. The subscript $\times$ for a comonoid and $+$ for a monoid is inspired by the behavior of products and coproducts in category theory: an object together with its diagonal and terminal morphisms forms a comonoid, and an object together with its codiagonal and initial morphisms forms a monoid.

We require morphisms of internal monoids to preserve the multiplication and the unit, just as for homomorphisms of conventional monoids. Given a symmetric monoidal category \cat{C}, and two internal commutative monoids $(A,g,u)_+$ and $(B,h,v)_+$, a morphism $m:A \to B$ is a morphism of comonoids if and only if the following diagram commutes:
\begin{equation}
\begin{diagram}[width=50pt,height=20pt]
A\otimes A & \rTo^{m \otimes m} & B\otimes B \\
\dTo<g&&\dTo>h \\
A & \rTo^m & B \\
\uTo<u &&\uTo>v \\
I & \lEq & I \\
\end{diagram}
\vspace{5pt}
\end{equation}
The dual definition, of a morphism of comonoids, is obtained by reversing all the arrows in this diagram. We use the notion of morphism of monoids to define \cat{C_+}, the category of internal commutative monoids in the symmetric monoidal category \cat{C}, which has internal commutative monoids in \cat{C} as objects and morphisms of monoids as arrows. We can similarly construct \cat{C_\times}, the category of internal cocommutative comonoids in \cat{C}.

We now look at this in the context of all the structure that we have developed, and consider a category \cat{C} with the structure of a $\dag$-category with biproducts. We can make the following observations about internal monoids and comonoids in \cat{C}, all of which have reasonably simple proofs based on the contents of this section.

Firstly, using the $\dag$-structure, it is clear that every monoid $(A,g,u)_+$ gives rise to a comonoid \mbox{$(A,g^\dag,u^\dag) _\times$}. If \mbox{$m:(A,g,u)_+ \to (B,h,v)_+$} is a morphism of monoids, then the adjoint morphism ${m^\dag} : (B,h^{\dag}, {v^\dag}) _\times \to (A,g^\dag,u^\dag)_\times$ is a morphism of comonoids. It follows that the categories \cat{C_\times} and \cat{C_+} are opposite to each other.

We now consider the symmetric monoidal structure. We can use the structural isomorphisms to define a commutative comonoid on the monoidal unit:
\begin{equation}
\quad I_\times := (I, \lambda_I^{-1}, \id{I}) _\times.
\end{equation}
For any comonoid $(A,g,u)_\times$, the only morphism of comonoids $(A,g,u) _\times \to I_\times$ is $u$ itself; in other words, $I_\times$ is the terminal object in \cat{C_\times}. In fact, the category \cat{C_\times} has finite products, with binary product and projections defined as follows:
\begin{align*}
(A,g,u)_\times \times (B,h,v)_\times &\simeq \big( A \otimes B, (\id{A} \otimes \swap ^\otimes _{A,B} \otimes \id{B}) \circ (g \otimes h) , u \otimes v \big) _\times
\\
p _{(A,g,u) _\times} &= \id{A} \otimes v
\\p _{(B,h,v) _\times} &= u \otimes \id{B}
\end{align*}
The biproduct structure of \cat{C} gives rise to an initial object in \cat{C_\times}. We can define a comonoid on the zero object in the following way:
\begin{equation}
\quad 0_\times := (0, 0_{0,0\otimes 0}, 0_{0,I})_\times .
\end{equation}
For any comonoid $(A,g,u)_\times$, there is only one morphism of comonoids $0 _\times \to (A,g,u) _\times$, given by the zero morphism $0_{0,A}$, and so $0_\times$ is the initial object in \cat{C_\times}. Also, there is no morphism $(A,g,u)_\times \to 0_\times$ unless $(A,g,u) _\times$ is isomorphic to $0_\times$. If the tensor product distributes naturally over the biproduct in \cat{C}, then $\cat{C} _\times$ will in fact have coproducts, but we will not define these here as we do not need them for this paper.

The category $\cat{C} _\times$ therefore bears a clear resemblance to the structure of many categories of spaces, such as the category of sets: it is a category with distinct products and coproducts, such that the only object with a morphism to the initial object is the initial object itself. This gives us the motivation to think of \cat{C_\times} as a category of spaces, and for all comonoids $(A,g,u) _\times$, to think of morphisms $I _\times \to (A,g,u) _\times$ in $\cat{C} _\times$ as representing its \emph{points}. In the special case of commutative comonoids which are dual to commutative C*-algebras, this is a well-used construction: the points of the comonoid are precisely the elements of the spectrum of the C*-algebra\footnote{We note that internal monoids in \cat{FdHilb} will not, in general, give rise to C*-algebras: elements of the monoid lack a canonical involution, and the Banach algebra condition $||a \,b|| \leq ||a|| \, ||b||$ will not necessarily be satisfied. In fact, the Banach algebra condition seems quite unnatural in this context.}.

If we can find both a comonoid structure $(A,g,u) _\times$ and a monoid structure $(A,h,v) _+$ on the same object $A$, then we can consider the compatibility of the comonoid morphisms $g$ and $u$ with the monoid morphisms $h$ and $v$. If $g$ and $u$ are both morphisms of monoids, then we say that the combined structure $(A,g,u,h,v) _{\times +}$ is a \emph{bialgebra}. This requires the definition of a monoid on $A \otimes A$; we choose this monoid to have unit given by $v \otimes v$, and multiplication given by $(h \otimes h) \circ (\id{A} \otimes \swap_A ^\otimes \otimes \id{A})$. (This is equivalent to a similar specification of a comonoid on $A \otimes A$, and requiring that $h$ and $v$ be morphisms of comonoids.) Diagrammatically, the compatibility conditions that arise are as follows:
\begin{align*}
\begin{diagram}[midshaft,height=25pt,width=50pt]
A \otimes A \otimes A \otimes A & \rTo^{ \hspace{-50pt} \id{A} \otimes \swap _A ^\otimes \otimes \id{A} \hspace{-50pt} } & & & A \otimes A \otimes A \otimes A
\\
\uTo <{g \otimes g} & & & & \dTo>{h \otimes h}
\\
A \otimes A & \rTo^h & A & \rTo^g & A \otimes A
\end{diagram}
&&
\begin{diagram}[midshaft,height=25pt,width=25pt]
A & \rTo^u & I
\\
\uTo<v & \ruEq &
\\
I & &
\end{diagram}
\end{align*}
\begin{align*}
\begin{diagram}[midshaft,height=25pt,width=30pt]
A & \rTo^g & A \otimes A
\\
\uTo< v & & \uTo>{v \otimes v}
\\
I & \rTo _\simeq & I \otimes I
\end{diagram}
&&
\begin{diagram}[midshaft,height=25pt,width=30pt]
A & \lTo^h & A \otimes A
\\
\dTo< u & & \dTo>{u \otimes u}
\\
I & \rTo _\simeq & I \otimes I
\end{diagram}
\end{align*}

\section{The conventional quantum harmonic oscillator}

\label{conventionalqho}
The quantum harmonic oscillator is one of the simplest quantum systems that can be studied. It is also one of the most important; in particular, quantum field theory can be interpreted as a perturbation on top of an infinite-dimensional harmonic oscillator. The state space of the conventional quantum harmonic oscillator is known as Fock space, a Hilbert space of countably-infinite dimension. We will examine it in detail in section \ref{conventional}, to see how it arises as a special case of the more general categorical description. In this section we will give a brief introduction to the necessary physics, focusing on giving an intuitive description of the state space of the quantum harmonic oscillator and the tools that physicists use to work with it.

The quantum harmonic oscillator is the name given to the mathematical model describing a quantum particle trapped in a quadratic potential. Energy levels of this system are not continuous, as they would be for the classical harmonic oscillator; they are discrete, given by $E_n := hf (n + \frac{1}{2})$ for all natural numbers $n \geq 0$, where $h$ is Planck's constant and $f$ is the characteristic frequency of the system. In particular, we note that the $n=0$ state has non-zero energy, which is surprising from the point of view of classical mechanics. Specifying a state of the quantum harmonic oscillator amounts to specifying an amplitude for the system to be found in any of these energy levels. The state space is clearly countably-infinite dimensional, since there is a countable number of energy levels for the system.

Inspired by quantum field theory, we will refer to the different energy levels of each harmonic oscillator as counting numbers of \emph{particles}; for example, a harmonic oscillator with energy $E_3$ is in a three-particle state. The lowest energy state is the zero-particle state, or the vacuum. This is a useful language, but perhaps one that should not be taken too seriously without better motivation. We will explore the particle interpretation further in section \ref{pathcounting}.

We shall now describe \emph{symmetric Fock space}. This is the state space of a compound system formed from a countable number of harmonic oscillators, but with the extra requirement that permutation of particles is a symmetry of the state space; our particles are symmetrically distinguishable\footnote{The alternative, which we do not consider here, is that the particles be antisymmetrically indistinguishable, which implies a change of sign under every permutation.}. So, given the symmetric Fock space over two harmonic oscillators, a possible state is one in which there are two particles on the first oscillator; that is, the first oscillator has energy $E_2$ and the second oscillator has energy $E_0$. But there is no state in the space for which a first particle is on the first oscillator and a second particle is on the second oscillator, since such a state is not invariant under interchange of the particles. There \emph{is} a state with two particles for which each oscillator has energy $E_1$, but no information can be gathered as to which particle is where.

In general, a state of symmetric Fock space is fully described by specifying, separately for each possible total number of particles, a complex amplitude for each way that this number of particles could be distributed between all of the available oscillators. Note that such a specification contains no information about which particle is where, and so respects the symmetric indistinguishability criterion. If we write $A$ for the single-particle subspace of this state space, sufficient to describe the possible ways that a single particle could be distributed among the oscillators, then we see that the total symmetric Fock space, written $F(A)$, has the following structure:
\begin{align*}
F(A) &= \mathbb{C} \oplus A \oplus (A \! \otimes_s \! A) \oplus (A\! \otimes_s\! A\! \otimes_s \! A) \oplus \cdots.
\end{align*}
 The symbol $\oplus$ represents disjoint union (or direct sum) of state spaces, and the symbol $\otimes _s$ represents symmetric tensor product. We refer to each $n$-fold symmetric tensor product of $A$ as the $n$-particle subspace of $F(A)$. The zero-particle subspace $\mathbb{C}$ consists of a single complex number, since if there are no particles, there is no information to give about how these particles are arranged; one need only give the amplitude that this is the case.

To work with symmetric Fock space, physicists have developed a number of tools. The most important are the raising and lowering linear operators, $a_\phi ^\dag : F(A) \to F(A)$ and $a _\phi: F(A) \to F(A)$ respectively, sometimes known as the creation and annihilation operators. Here, $\phi$ is a state of $A$, the single-particle subspace. Applying the raising operator $a_\phi ^\dag$ to an $n$-particle state of $F(A)$ --- that is, a state which is zero except in the $n$-particle subspace --- creates an $(n+1)$-particle state of $F(A)$, with the same particle content as before, except for the addition of a new particle described by the state $\phi$. The lowering operator $a_\phi$ performs the adjoint to this process, turning an $(n+1)$-particle state into an $n$-particle state by removing a particle in the state $\phi$. If there was no amplitude to find a particle in the state $\phi$ in the first place, then applying $a_ \phi$ will annihilate the state, giving zero. Also, applying the lowering operator to the zero-particle state will give zero.

These raising and lowering operators satisfy various commutation relations, describing the different effects created by applying them in different orders. If we add two particles with the raising operators $a _\phi ^\dag$ and $a _\psi ^\dag$, it should not matter in which order we choose to apply them, and similarly for the case of two different lowering operators. So we expect the following commutation relations to hold for all $\phi$ and $\psi$:
\begin{align*}
\big[a _\psi, a_\phi \big] &\subset 0
\\
\big[a _\psi ^\dag, a_\phi ^\dag \big] &\subset 0
\end{align*}
The subset symbols indicate equality on the domain of the left-hand side of the equation, since the raising and lowering operators, being unbounded, will not be everywhere-defined. However, if we apply $a _\phi$ and $a _\psi ^\dag$, it is not clear that the order should be unimportant. In fact, the physics of the harmonic oscillator tells us that the commutation relation should be as follows:
\begin{equation*}
\big[a ^{\phantom{\dag}}_\phi, a_\psi ^\dag \big] \subset (\phi,\psi) \, \id{}
\end{equation*}
Here, $(\phi, \psi)$ is the inner product of the two single-particle states, linear in $\psi$ and antilinear in $\psi$, and $\id{}$ is the identity on the Hilbert space. The raising and lowering operators are obtained as a combination of the position and momentum operators, and this commutation relation arises directly from Heisenberg's uncertainty relation, which states that measurements of position and momentum do not commute. An intriguingly direct argument is presented in \cite{mortonspecies}, where this commutation relation is related to the fact that there is one more way to first put a ball into a box and then take a ball out, then there is to first take a ball out of a box and then put one in.

Finally, we describe a family states of symmetric Fock space known as the coherent states. These are parameterised by the single-particle states; for each state $\phi$ of $A$, we write $\Coh(\phi)$ for the corresponding coherent state of $F(A)$. We can construct $\Coh(\phi)$ explicitly in the following way:
\begin{equation*}
\Coh(\phi) := 1 \oplus \phi \oplus \left( \frac {1} {\sqrt{2!}} \,\phi \! \otimes_s \! \phi \right)\oplus \left( \frac {1} {\sqrt{3!}} \, \phi \! \otimes_s \! \phi \! \otimes_s \! \phi \right) \oplus \cdots
\end{equation*}
These states always have finite norm; we have \mbox{$||\Coh(\phi)|| ^2= e ^{||\phi|| ^2}$}. Physically, coherent states contain an indeterminate number of particles, all of which are in the same state. There are other states with these properties that are not of this form, but the defining feature of the coherent states are that they satisfy the following equation, for all single-particle states $\phi$ and $\psi$:
\begin{equation*}
a_\phi \, \Coh(\psi) = (\phi,\psi)\, \Coh(\psi)
\end{equation*}
The coherent states are eigenstates for the lowering operators; intuitively, removing a particle from a coherent state only modifies the state by a factor. Two more interesting properties enjoyed by the coherent states is that they can be copied and deleted: there exists linear maps $d:F(A) \to F(A) \otimes F(A)$ and $e:F(A) \to \mathbb{C}$ such that
\begin{align*}
d\,\Coh(\phi) &= \Coh(\phi) \otimes \Coh(\phi)
\\
e \, \Coh(\phi) &= 1
\end{align*}
for all single-particle states $\phi$. For these reasons, the coherent states are often thought of as having classical properties. Finally, we note that the coherent states in the form given can be constructed using the following identity, where we define $v$, the `vacuum state', to be the state in which there is an amplitude of $1$ to find zero particles, and no amplitude to find any greater number:
\begin{equation*}
\exp(a ^\dag _\phi) \, v = \Coh(\phi)
\end{equation*}
This provides an interesting connection between the zero-particle state, the raising operators, the coherent states and the exponential function.

This completes our tour of the classical treatment of the quantum harmonic oscillator. It would seem that the most crucial part is the structure of symmetric Fock space; the rest of the structure could be regarded merely as tools developed by physicists to aid its study. However, we shall see that the categorical approach efficiently reproduces (and indeed generalises) all of the structures and techniques described in this section: not only symmetric Fock space, but also the zero- and single-particle subspaces, the raising and lowering operators and their commutation relations, the coherent states and their copying and deleting maps, the exponential function, and all the equations relating them which we have described.

\section{A categorical description of the quantum\\ harmonic oscillator}
\label{formulation}

\subsection{The categorical framework}
We begin with our categorical description of the quantum harmonic oscillator.

\begin{defn}
\label{defhoa}
Given a symmetric monoidal $\dag$-category \cat{C} with finite $\dag$-biproducts, an \emph{harmonic oscillator adjunction} $\llangle Q , \eta , \epsilon \rrangle$ is a right adjoint $Q : \cat{C} \to \cat{C} _\times$ for the forgetful functor \mbox{$R : \cat{C}_\times \to \cat{C}$}, with unit $\eta : \id{\cat{C} _\times} \dot{\to} Q \circ R$ and counit $\epsilon : R \circ Q \dot{\to} \id{\cat{C}}$, that has the following properties:
\begin{enumerate}
\item The functor $Q$ preserves finite products unitarily;
\item The natural transformation $\epsilon ^\dag$ is an isometry at every stage; that is, $\epsilon \circ \epsilon ^\dag = \id{\id{\cat{C}}}$;
\item The endofunctors $R \circ Q : \cat{C} \to \cat{C}$ and $\dag: \cat{C} \to \cat{C}$ commute.
\end{enumerate}
\end{defn}

\noindent For each object $A$ in \cat{C}, we interpret $Q(A)$ in $\cat{C} _\times$ as the \emph{harmonic oscillator} constructed over $A$. Defining the endofunctor $F: \cat{C} \to \cat{C}$ as
\begin{equation}
\quad F := R \circ Q ,
\end{equation}
we interpret $F(A)$ as the state space of the comonoid $Q(A)$, obtained by `forgetting' the comultiplication and counit morphisms. States of $Q(A)$ are therefore given by non-zero morphisms \mbox{$\phi : I \to F(A)$} in \cat{C}, following the general framework of categorical quantum mechanics as discussed in section \ref{symmon}. $F$ can be thought of a generalised Fock space functor, or a generalised `second quantisation' functor. We also see that the adjunction $R \dashv Q$ gives $F$ the structure of a comonad\footnote{This allows us to update the old adage: ``First quantization is a mystery, but second quantization is a comonad.''}.

A related approach, developed in parallel to this work by another author \cite{differentialfiore}, is to consider the comonad $(F,\epsilon, R \eta Q)$ as primary rather than the adjunction. This is a more general framework, but one in which the counit morphisms $R \eta _{(A,g,u) _\times}$ will not be available for all commutative comonoids $(A,g,u) _\times$\footnote{In terms of the comonad rather than the adjunction, morphisms corresponding to the $R \eta _{(A,g,u) _\times}$ arise as \emph{coalgebras} for the comonad. If there exists a coalgebra for every comonoid, then the category of coalgebras (the co-Eilenberg-Moore category) will be equivalent to the category of comonoids, and the comonad must arise from an adjunction in precisely the sense of definition \ref{defhoa}.}. We will make substantial use of these morphisms later in the paper, and so the current construction is more convenient for our purposes.

To work with the comultiplication and counit of each categorical harmonic oscillator $Q(A)$ more easily, we make the following definition for the remainder of the paper:
\begin{equation}
\big (F(A), d_A, e_A \big) _\times := Q(A).
\end{equation}
We will explore in the next few sections of this paper how this structure, along with the unit and counit natural transformations, endow $Q(A)$ with many of the properties of a conventional quantum harmonic oscillator.

Before we go any further, we must make clear what it means in definition \ref{defhoa} for $Q$ to preserve products unitarily. Products in \cat{C} are given by the $\dag$-biproduct structure, and in $\cat{C} _\times$ they are given by the underlying tensor product in \cat{C}, as discussed in section \ref{tools}:
\begin{equation*}
\quad (A,g,u) _\times \times (B,h,v) _\times \simeq \big(A \otimes B, (\id{A} \otimes \swap ^\otimes _{A,B} \otimes \id{A}) \circ (g \otimes h), u \otimes v \big) _\times .
\end{equation*}
As $Q$ is a right adjoint it must preserve finite products up to unique natural isomorphism, and therefore for all $A$ and $B$ in \cat{C} there exist in $\cat{C} _\times$ unique natural isomorphisms
\begin{align}
k_{A,B} &: Q(A \oplus B) \to Q(A) \times Q(B)
\\
k_0 &: Q(0) \to I _\times
\end{align}
which make the following diagrams commute:
\begin{equation*}
\begin{diagram}[midshaft,height=20pt]
Q(A) & \lTo^{Q( \pi _A) } & Q(A \oplus B) & \rTo ^{Q (\pi _B)} & Q(B)
\\
\dEq & & \uTo < {k _{A,B} ^{-1}} \dTo >{k _{A,B}} & & \dEq
\\
Q(A) & \lTo _{ \id{Q(A)} \times e_B} & Q(A) \times Q(B) & \rTo _{ e_A \times \id{Q(B)}} & Q(B)
\end{diagram}
\qquad\quad
\begin{diagram}[midshaft,height=20pt]
Q(0)
\\
\uTo <{k_0 ^{-1}} \dTo >{k_0}
\\
I_ \times
\end{diagram}
\end{equation*}
We include the trivial right-hand diagram for completeness. For $Q$ to preserve products unitarily means that the natural isomorphisms $k_{A,B}$ and $k_0$ are unitary, when viewed as morphisms in \cat{C}\ using the forgetful functor $R$:
\begin{align*}
R k_{A,B} ^{-1} = (R k_{A,B}) ^\dag
&&
R k_{0} ^{-1} = (R k_0) ^\dag
\end{align*}

We can use the natural isomorphisms $k_{A,B}$ and $k_0$ to obtain explicit expressions for the comultiplication $d_A$ and counit $e_A$ associated to each categorical harmonic oscillator $Q(A)$.
\begin{lemma}
\label{autoed}
For all objects $A$ in \cat{C}, we have
\begin{align}
\label{eeqn}
e_A &= Rk_0 \circ F (0 _{A,0} )
\\
\label{deqn}
d_A &= Rk_{A,A} \circ F( \Delta_A )
\end{align}
where we view $e_A$ and $d_A$ as morphisms in \cat{C} rather than $\cat{C} _\times$.
\end{lemma}
\begin{proof}
Since terminal morphisms are unique, and terminal objects are preserved by $Q$, equation (\ref{eeqn}) must hold automatically.
The unit law for $Q(A)$ can be reinterpreted as the expression that $d_A$ is the diagonal for the object $Q(A)$ in $\cat{C} _\times$; in other words, $d_A = \langle \id{Q(A)} ,\id{Q(A)} \rangle$ in $\cat{C} _\times$. It must therefore be related by $c_{A,A}$ to the image under $Q$ of the diagonal $\Delta_A := \langle \id{A}, \id{A} \rangle$ in \cat{C}, implying equation (\ref{deqn}).
\end{proof}

In fact, the commutative comonoids in the image of the functor $Q$ have more structure than is immediately apparent; as discussed in \cite{differentialfiore}, they are bialgebras, for which the multiplication and unit morphisms are given by the adjoint of the comultiplication and counit in \cat{C}.
\begin{lemma}[Fiore]
\label{bialgebralemma}
For each object $A$ in \cat{C}, $\big(F(A), d_A ^{\phantom{\dag}}, e_A ^{\phantom{\dag}}, d_A ^\dag, e_A ^\dag \big) _{\times +}$ is a bialgebra.
\end{lemma}
\begin{proof} In a category with biproducts, the diagonal and codiagonal along with terminal and initial morphisms form a bialgebra with respect to the monoidal structure of the biproduct, as can be checked by working through all of the necessary diagrams. This bialgebraic structure is inherited by $d_A$, $e_A$ and their adjoints, since they are formed `naturally' from the biproduct structure in \cat{C}, the unitary isomorphisms $k_{A,B}$ and $k_0$ translating between the biproduct and tensor product structures.
\end{proof}

\noindent
We can also prove that the comultiplication $d_A$ is \emph{additive}, in the following sense.
\begin{lemma}
\label{additivity}
For all morphisms $f$, $g:A \to B$ in \cat{C},
\begin{equation}
\quad F(f +g) = d_B ^\dag \circ ( F(f) \otimes F(g)) \circ d_A.
\end{equation}
where the sum $f+g$ is defined by the $\dag$-biproduct structure.
\end{lemma}
\begin{proof}
Using using naturality of $c_{A,A}$ and compatibility of $F$ with $(-) ^\dag$, we obtain:
\begin{equation*}
\mbox{
\begin{diagram}[midshaft,loose,width=0pt,height=15pt,objectstyle=\scriptstyle,labelstyle=\scriptstyle]
F(A) & \rTo^{F( \Delta _A)} & F(A \oplus A) & \rTo^{Rk_{A,A}} & F(A) \otimes F(A) & \rTo^{F(f) \otimes F(g)} & F(B) \otimes F(B) & \rTo^{(Rk_{B,B}) ^\dag} & F(B \oplus B) & \rTo^{F(\Delta _B) ^\dag} & F(B)
\\
& & \dEq & & & & \dEq & & \dEq & & \dEq
\\
& & F(A \oplus A) & \rTo^{K(f \oplus g)} & F(B \oplus B) & \rTo^{Rk_{B,B}} & F(B) \otimes F(B) & & F(B \oplus B) & \rTo^{F(\nabla _B) } & F(B)
\\
& & & & \dEq & & & & \dEq
\\
& & & & F(B \oplus B) & \rEq & & & F(B \oplus B)
\\
\dEq & & & & & & & & & & \dEq
\\
F(A) & \rTo^{\hspace{-15pt}F(f+g)\hspace{-15pt}} & & & & & & & & & F(B)
\end{diagram}
}
\end{equation*}
\end{proof}

Although the natural transformations $e$ and $\epsilon$ are not directly related by any equations in the construction of the harmonic oscillator adjunction, they nevertheless automatically satisfy some compatibility conditions.
\begin{lemma}
\label{orthognormlemma}
The natural transformations $\epsilon$ and $e$ are normalised and orthogonal; that is, at every stage $A$,
\begin{align}
\label{epseps}
\epsilon_A ^{\phantom{\dag}} \circ \epsilon_A ^\dag &= \id{A}
\\
\label{ee}
e_A ^{\phantom{\dag}}\circ e_A^\dag &= \id{I}
\\
\label{eeps}
e_A ^{\phantom{\dag}} \circ \epsilon_A^\dag &= 0_{A,I}
\end{align}
\end{lemma}
\begin{proof}
Equation (\ref{epseps}) holds by construction, since it was a requirement in definition \ref{defhoa} of a harmonic oscillator adjunction. To tackle equation (\ref{ee}), we use lemma \ref{autoed} to write $e_A$ in terms of $k_0$; we then obtain
\begin{align*}
e_A \circ e_A ^\dag &= Rk_0 \circ F(0_{A,0}) \circ F(0 _{0,A}) \circ (R k_0) ^\dag
\\
&= Rk_0 \circ F(0_{0,0}) \circ (R k_0) ^\dag
\\
&= Rk_0 \circ \id{F(0)} \circ (R k_0) ^\dag
\\
&= \id{I},
\end{align*}
where we have employed unitarity of $R k_0$. For equation (\ref{eeps}), we have
\begin{align*}
e _A \circ \epsilon _A ^\dag &= R k_0 \circ F(0_{A,0}) \circ \epsilon _A ^\dag
\\
&= R k_0 \circ \epsilon_0 ^\dag \circ 0_{A,0}
\\
&= 0_{A,I},
\end{align*}
where we use naturality of $\epsilon ^\dag$ and the fact that only zero morphisms can factor through zero morphisms.
\end{proof}
\noindent
In fact, as is well-known in other contexts \cite{llfock}, we can always write $k_{A,B}$ directly in terms of $e$ and $\epsilon$ in the following way.
\begin{lemma}
\label{gettingk}
Under the canonical hom-set isomorphism
\begin{equation*}
H _{Q(A) \times Q(B), A \oplus B}: \Hom_{\cat{C}} \big( F(A)\otimes F(B), A \oplus B \big) \simeq \Hom_{\cat{C} _\times} \big( Q(A) \times Q(B), Q(A \oplus B) \big)
\end{equation*}
induced by the adjunction, the family of morphisms in \cat{C} given by
\begin{equation*}
r_{A,B} := i_A\epsilon_A \otimes e_B + e_A \otimes i_B \epsilon_B,
\end{equation*}
where $i_A$ and $i_B$ are canonical injections into the biproduct $A \oplus B$, produce the morphisms $k_{A,B} ^{-1}$ in $\cat{C} _\times$; that is, for all $A$ and $B$ in \cat{C},
\begin{equation}
k _{A,B} ^{-1} = H_{Q(A) \times Q(B), A \oplus B} (r_{A,B}).
\end{equation}
\end{lemma}
\begin{proof}
The morphisms $H_{Q(A) \times Q(B), A \oplus B} (r_{A,B})$ defined in this manner are clearly well-defined morphisms of comonoids. We must show that it mediates between the product structures in $\cat{C} _\times$; it suffices to show that for any $A$ and $B$, $Q(i_A ^\dag) \circ k_{A,B} ^{-1}  = \id{Q(A)} \times e_B$. These are morphisms in $\Hom_{\cat{C} _\times} (Q(A) \times Q(B),Q(A))$, and we apply the hom-set isomorphism once again to view them as morphisms in $\Hom_{\cat{C}} (F(A) \otimes F(B), A)$. We obtain
\begin{align*}
& \hspace{-30pt}\epsilon_A \circ F(i_A ^\dag) \circ R(H_{Q(A) \times Q(B), A \oplus B} (r_{A,B}))
\\
&= \epsilon_A \circ F(i_A ^\dag) \circ F(i_A \epsilon_A \otimes e_B + e_A \otimes i_B \epsilon_B) \circ R\eta _{Q(A) \times Q(B)}
\\
&= \epsilon_A \circ F(\epsilon_A \otimes e_B) \circ R\eta _{Q(A) \times Q(B)}
\\
&=(\epsilon_A \otimes e_B) \circ \epsilon _{F(F(A) \otimes F(B))} \circ R \eta _{Q(A) \times Q(B)}
\\
&= \epsilon_A \otimes e_B.
\end{align*}
But this is equal to $\epsilon_A \circ R(\id{Q(A)} \times e_B)$, and so the product-preservation equation holds.
\end{proof}

\noindent We can use this to prove another very useful result.
\begin{lemma}
\label{edexpand}
At any stage $A$, we can write $\epsilon_A \circ d_A ^\dag : F(A) \otimes F(A) \to A$ as
\begin{align}
\epsilon_A \circ d _A ^\dag &= \epsilon_A \otimes e_A + e_A \otimes \epsilon_A.
\end{align}
\end{lemma}
\begin{proof} Using lemmas \ref{autoed} and \ref{gettingk}, along with naturality of $\epsilon$ and one of the adjunction equations, we obtain
\begin{align*}
\epsilon_A \circ d_A ^\dag &= \epsilon_A \circ F( \nabla_A) \circ F(i_A \epsilon_A \otimes e_A + e_A \otimes i_A \epsilon_A) \circ R \eta _{Q(A) \times Q(A)}
\\
&= \epsilon_A \circ F(\epsilon_A \otimes e_A + e_A \otimes \epsilon_A) \circ R \eta _{Q(A) \times Q(A)}
\\
&= (\epsilon_A \otimes e_A + e_A \otimes \epsilon_A) \circ \epsilon_{F(F(A) \otimes F(A))} \circ R \eta _{Q(A) \times Q(A)}
\\
&= \epsilon_A \otimes e_A + e_A \otimes \epsilon_A.
\qedhere
\end{align*}
\end{proof}
\subsection{Graphical representation}

We will develop a graphical representation for the extra structure associated with a categorical harmonic oscillator, which will help us to prove theorems more easily. The most basic structure is the functor $F: \cat{C} \to \cat{C}$. We represent it as a pair of dashed lines, one on each side of its argument, for its action on both objects and morphisms. Functoriality of $F$ then means that graphical components within the dashed lines can be manipulated as if the dashed lines were not there; $F$ acts as an `inert container'. This principle is illustrated in the following diagram, which holds for all $f:A \to B$ and $g:B \to C$ in \cat{C}:
\begin{equation*}
\psfrag{RQA}{\hspace{-55pt}$F(A)$}
\psfrag{RQC}{\hspace{-55pt}$F(C)$}
\psfrag{RQB}{\hspace{-102pt}$F(B)$}
\psfrag{f}{\hspace{-8pt}\raisebox{-1pt}{$f$}}
\psfrag{g}{\hspace{-5pt}\raisebox{0pt}{$g$}}
\psfrag{gcf}{\hspace{-12pt}\raisebox{-1pt}{$g \circ f$}}
\begin{array}{c}
\hspace{27pt}
\includegraphics{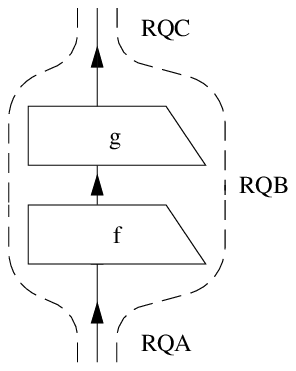}
\hspace{-25pt}
\end{array}
=
\begin{array}{c}
\includegraphics{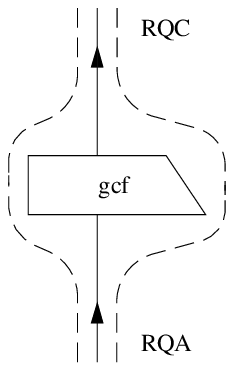}
\end{array}
\end{equation*}

To represent the comultiplication $d_A$ and counit $e_A$ graphically for each comonoid $Q(A)$, we extend the graphical representation for comonoids which we developed in section \ref{tools}. The graphical components that we will use are as follows, defined for all objects $A$ in \cat{C}:
\secretlabel{drep}
\secretlabel{erep}
\begin{equation*}
\begin{array}{c}
\begin{array}{c}
\includegraphics{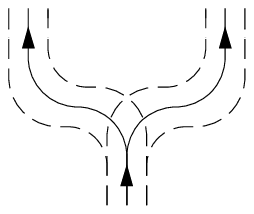}
\end{array}
\\
d_A:F(A) \to F(A) \otimes F(A)
\end{array}
\hspace{50pt}
\begin{array}{c}
\begin{array}{c}
\includegraphics{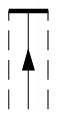}
\end{array}
\\
e_A : F(A) \to I
\end{array}
\end{equation*}
As we have done here, we will often not annotate our diagrams whenever this does not introduce ambiguities. We will also frequently work with the adjoints $d_A ^\dag$ and $e_A ^\dag$ without remark; following the conventions laid out in section \ref{tools}, these are given by flipping the diagrams along a horizontal axis, and then reversing the orientation of the arrows. Of course, these graphical components obey the dual versions of the associativity, unit and commutativity laws (\ref{assoc}), (\ref{unit}) and (\ref{comm}).

Finally, we introduce the representations for the unit and counit natural transformations. For each stage
\begin{align*}
\eta_{(A,g,u) _\times} &: (A,g,u) _\times \to Q(A),
\\
\epsilon_A &:F(A) \to A,
\end{align*}
we employ the following diagrams:
\begin{equation*}
\begin{array}{c}
\psfrag{ec}{$R\eta_{(A,g,u)_\times}$}
\includegraphics{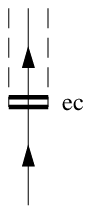}
\hspace{-38pt}
\\
R\eta_{(A,g,u) _\times}:A \to F(A)
\end{array}
\hspace{50pt}
\begin{array}{c}
\includegraphics{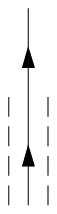}
\\
\epsilon_A : F(A) \to A
\end{array}
\end{equation*}
The graphical representation of $\epsilon_A$ features an unblocked `mouth' for the dashed lines. This is an intuitive way to represent of the naturality of $\epsilon$, which can be written algebraically as \mbox{$\epsilon_B \circ F(f) = f \circ \epsilon_A$} for any $f:A \to B$. In other words, graphical components can be freely moved across $e_A$, into and out of the functor $F$, traversing the `mouth' of the dashed lines:
\begin{equation*}
\psfrag{f}{\hspace{-8pt}\raisebox{-1pt}{$f$}}
\centregraphics{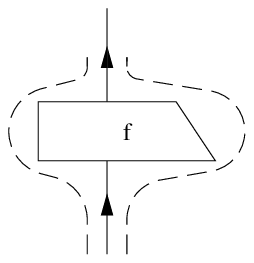}
=
\centregraphics{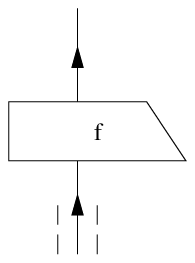}
\end{equation*}
For $R\eta_{(A,g,u)_\times}$, however, we mark the end of the functor $F$ by a double line; this is `harder' for graphical components to cross, as naturality of $\eta$ implies that only morphisms of comonoids may pass. 

Finally, there are compatibility equations satisfied by the natural transformations $\eta$ and $\epsilon$ which define the adjunction. We summarise these here, along with their graphical representations.
\begin{eqnarray}
\begin{array}{c}
\label{Qadjeqn}
\begin{diagram}[midshaft,height=30pt,width=35pt]
Q & \rTo^{\eta Q} & Q \circ R \circ Q
\\
& \rdTo_{\id{Q}} & \dTo>{Q \epsilon}
\\
& & Q
\end{diagram}
\end{array}
&\hspace{100pt}&
\begin{array}{c}
{
\psfrag{eta}{\hspace{-77pt}$R\eta ^{} _{Q(A)}$}
\includegraphics{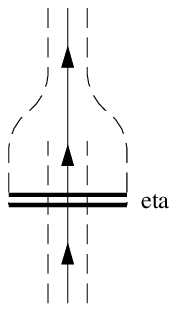}
\hspace{-20pt}
}
\end{array}
=
\begin{array}{c}
{
\includegraphics{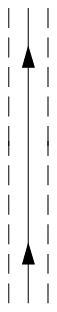}
}
\end{array}
\\
\label{Radjeqn}
\begin{array}{c}
\begin{diagram}[midshaft,height=30pt,width=35pt]
R & \rTo^{R \eta} & R \circ Q \circ R
\\
& \rdTo_{\id{R}} & \dTo>{\epsilon R}
\\
& & R
\end{diagram}
\end{array}
&\hspace{100pt}&
\begin{array}{c}
\hspace{3.5pt}
{
\psfrag{eta}{\hspace{-66pt}$R\eta_{(A,g,u) _\times}$}
\hspace{8pt}
\includegraphics{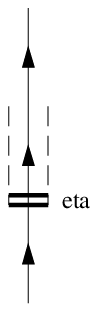}
\hspace{-1pt}
}
\hspace{-7pt}
\end{array}
=
\begin{array}{c}
\hspace{3.5pt}
{
\includegraphics{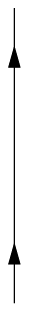}
}
\end{array}
\end{eqnarray}
These are the graphical representations  for arbitrary stages $A$ and $(A,g,u) _\times$ respectively. These adjunction equations are of course entirely category-theoretical, but they will prove essential to the physics, allowing us to demonstrate in theorem \ref{thmcohloweringeigenstate} that coherent states are eigenstates of the lowering operators, and in theorem \ref{cohasexp} that the coherent state can be written as the exponential of the raising operator.

\subsection{Raising and lowering morphisms}

In this section, we will construct generalised raising and lowering operators associated to the categorical harmonic oscillators.

For all morphisms $f : A \to B$ in $\cat{C}$, $Q(f):Q(A) \to Q(B)$ is a morphism of comonoids by the definition of $Q$, and so must satisfy the following diagram drawn in $\cat{C}$:
\begin{equation}
\begin{diagram}[width=50pt,height=20pt]
F(A) \otimes F(A) & \rTo^{F(f) \otimes F(f)} & F(B) \otimes F(B)
\\
\uTo <{d_A} && \uTo >{d_B} \\
F(A) & \rTo ^{F(f)} & F(B) \\
\dTo <{e_A} && \dTo >{e_B} \\
I & \rEq & I \\
\end{diagram}
\vspace{5pt}
\end{equation}
However, this is exactly the condition for $d_A$ and $e_A$ to be the stages of natural transformations
\begin{align*}
d &: F \to F \otimes F
\\
e &: F \to \cat{C}_I
\end{align*}
where $\cat{C}_I:\cat{C} \to \cat{C}$ is the functor that sends all objects in $\cat{C}$ to the monoidal identity object $I$, and all morphisms to $\id{I}$. We therefore have four basic natural transformations, arising from the comonoid structure and from the adjunction, which we can summarise using the following diagram:
\begin{equation}
\label{fourarrows}\vspace{0pt}
\begin{diagram}[width=40pt,height=30pt,midshaft,loose]
&& \hspace{0pt} F\otimes F \hspace{0pt} &\\
&\textbf{Comonoid}\,\,\,\,&\uTo<{d}&&\\
\cat{C}_I & \lTo^{e} & F & \rTo _{R \eta Q} & FF
\\
&&\dTo >{\epsilon}& \,\,\,\, \textbf{Adjunction} &
\\
&&\id{\cat{C}}&&
\end{diagram}
\end{equation}

We will define some new natural transformations, and use them to define the raising and lowering morphisms.
\begin{defn} The \emph{lowering natural transformation} $a: F \dot{\to} F \otimes \id{ \cat{C}}$ is defined as the following composite natural transformation:
\begin{equation}
\begin{diagram}[midshaft,width=40pt,height=15pt]
F & \rTo^a & & & F \otimes \id{ \cat{C}}
\\
\dEq{} & & & & {}\dEq
\\
F & \rTo^d & F \otimes F & \rTo^{\id{F} \otimes \epsilon} & F \otimes \id{\cat{C}}
\end{diagram}
\vspace{10pt}
\end{equation}
The \emph{raising natural transformation} $a^\dag:F \otimes \id{\cat{C}} \dot{\to} F$ is the adjoint to this.
\end{defn}

\begin{defn} The \emph{lowering morphism} $a_{\phi} : F(A) \to F(A)$ associated to the state $\phi : I \to F(A)$ is defined  as follows:
\begin{equation}
\label{lowering}
\begin{diagram}[midshaft,width=40pt,height=15pt]
F(A) & \rTo^{a_{\phi}} & & & F(A)
\\
\dEq{} & & & & {}\dEq
\\
F(A) & \rTo^{a_A} & F(A) \otimes A & \rTo^{\id{F(A)} \otimes \phi^\dag} & F(A)
\end{diagram}
\vspace{10pt}
\end{equation}
The \emph{raising morphism} $a^\dag _{\phi} : F(A) \to F(A)$ associated to $\phi$ is defined similarly:
\begin{equation}
\begin{diagram}[midshaft,width=40pt,height=15pt]
F(A) & \rTo^{a ^\dag _\phi} & & & F(A)
\\
\dEq &&&& \dEq
\\
F(A) & \rTo^{\id{F(A)} \otimes \phi} & F(A) \otimes A & \rTo ^{a ^\dag _A} & F(A)
\end{diagram}
\end{equation}
\end{defn}
\noindent The lowering morphism $a_{\phi}$ has a state as its subscript, and the stage $a_A$ of the lowering natural transformation has an object as its subscript, so they can be differentiated. Of course, the raising and lowering morphisms are related directly by the $\dag$ functor:
\begin{equation*}
\quad a_ \phi ^\dag \equiv (a_{\phi}) ^\dag.
\end{equation*}

The raising and lowering morphisms can perhaps be understood more intuitively by their graphical representations:
\begin{equation*}
\begin{tabular}{ccc}
{\psfrag{p}{$\phi$}
\includegraphics{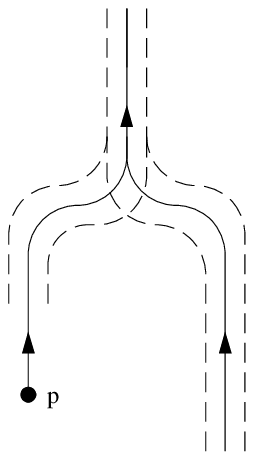}
}
&&
{\psfrag{pdag}{$\phi^\dag$}
\includegraphics{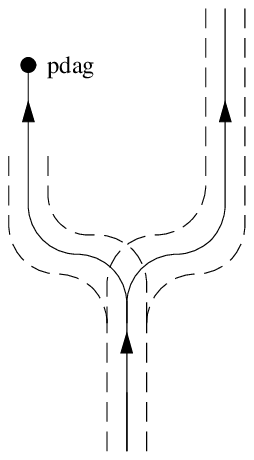}
}
\\
Raising morphism
&&
Lowering morphism
\\
$a ^\dag _\phi : F(A) \to F(A)$
&\qquad\qquad\qquad&
$a_{\phi} : F(A) \to F(A)$
\end{tabular}
\end{equation*}
We apply the raising morphism $a _\phi ^\dag$ to $F(A)$ by creating a new state $\phi$ of $A$, embedding it into $F(A)$, then multiplying with our original instance of $F(A)$. To lower with respect to $\phi$, we comultiply, then extract the single-particle state from one of the legs and take the inner product with $\phi$.

A cornerstone of the conventional analysis of the quantum harmonic oscillator is the set of canonical commutation relations, or CCRs, as described in section \ref{conventionalqho}:
\begin{equation*}
\begin{array}{ccccc}
\left[ a _{\phi}, a _{\psi} ^\dag \right] \subset ( \phi, \psi ) \,\id{}
&\quad&
\left[ a ^\dag_{\phi}, a ^\dag _{\psi} \right] \subset 0
&\quad&
\big[ a _{\phi}, a _{\psi} \big] \subset 0
\end{array}
\end{equation*}
We will show that appropriate categorical versions of these equations hold in our construction.

\begin{theorem}
For all objects $A$ in \cat{C}, and all states $\phi,\psi:I \to A$, the following commutation relations hold:
\begin{align}
\label{ccr}
a _{\phi} \circ a _{\psi} ^\dag &= a _{\psi} ^\dag \circ a _{\phi} + (\phi ^\dag \circ \psi) \cdot\id{F(A)}
\\
\label{ccrraising}
a ^\dag_{\phi} \circ a ^\dag _{\psi} &= a ^\dag _{\psi} \circ a ^\dag_{\phi}
\\
\label{ccrlowering}
a _{\phi} \circ a _{\psi} &= a _{\psi} \circ a _{\phi}
\end{align}
\end{theorem}
\noindent We first note that the domain issue in the conventional case is not relevant here; indeed, it is not even expressible. Secondly, we have rearranged the equations to some extent, as we will not in general have a way to express negatives in our category.
\begin{proof}
Given the definitions of the categorical raising and lowering morphisms, equations (\ref{ccrraising}) and (\ref{ccrlowering}) follow straightforwardly from cocommutativity and coassociativity of the comonoid $Q(A)$. We demonstrate this using the graphical representation for the $a ^\dag$ commutation relation; the proof for $a$ is analogous.    
\begin{gather*}
\begin{array}{c}
\begin{diagram}[height=30pt]
F(A)
\\
\uTo < {a _\phi ^\dag}
\\
F(A)
\\
\uTo < {a _\psi ^\dag}
\\
F(A)
\end{diagram}
\end{array}
=
{
\psfrag{f}{$\phi$}
\psfrag{p}{$\psi$}
\centregraphics{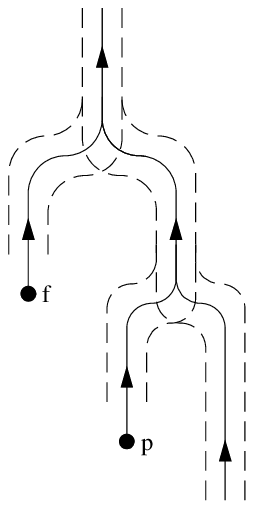}
=
\centregraphics{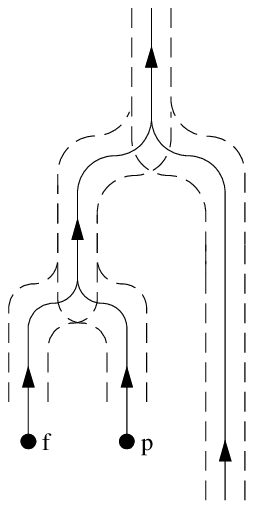}
}
\\
=
{
\psfrag{f}{$\psi$}
\psfrag{p}{$\phi$}
\centregraphics{graphics/adagcomm2}
=
\centregraphics{graphics/adagcomm1}
}
=
\begin{diagram}[height=30pt]
F(A)
\\
\uTo < {a _\psi ^\dag}
\\
F(A)
\\
\uTo < {a _\phi ^\dag}
\\
F(A)
\end{diagram}
\end{gather*}

\newpage
Our proof of equation (\ref{ccr}) is adapted from \cite{differentialfiore}. We employ lemmas \ref{bialgebralemma}, \ref{orthognormlemma} and \ref{edexpand}, and work at an arbitrary state $A$ in \cat{C}.
\begin{align*}
\psfrag{epsilonA}{}
\psfrag{epsilonAdagger}{}
a_A ^{\phantom{\dag}} \circ a_A^\dag
\quad
&=
\centregraphics{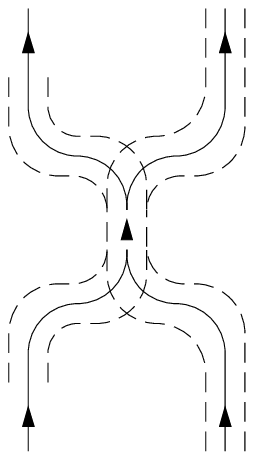}
=
\centregraphics{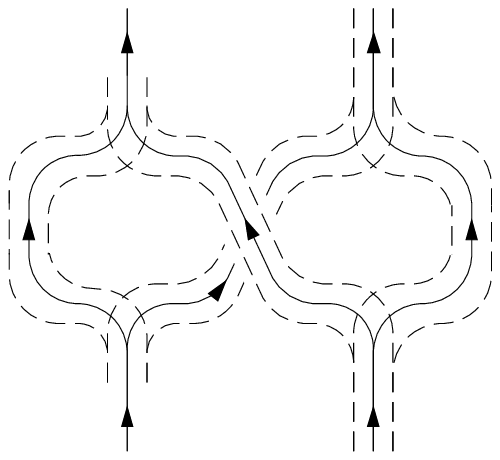}
\\
& =\centregraphics{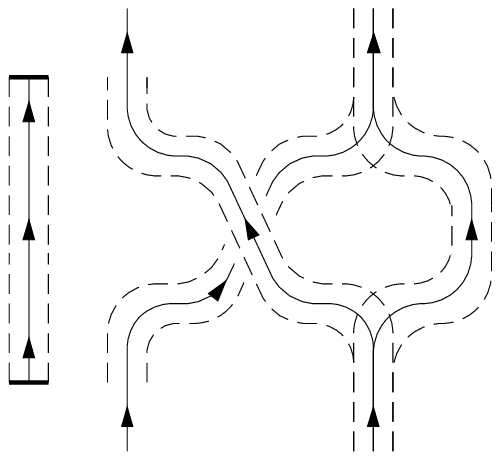}
+
\centregraphics{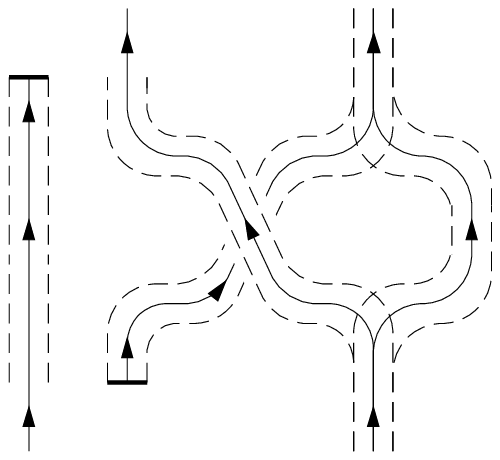}
\\
& \hspace{30pt}+ \centregraphics{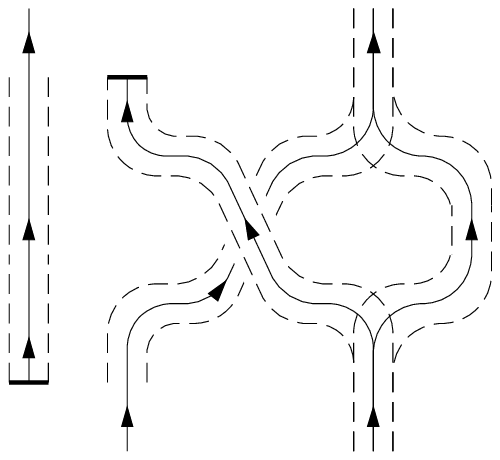}
+ \centregraphics{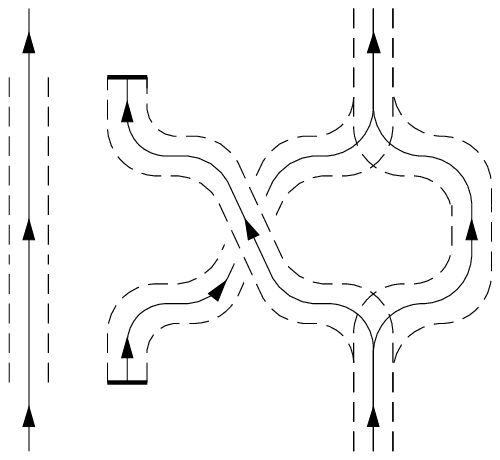}
\\
& =
\centregraphics{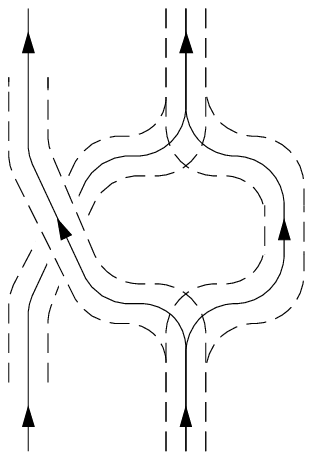}
+ \quad 0 \quad+ \quad 0 \quad+
\centregraphics{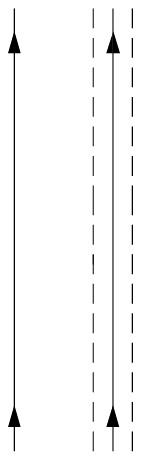}
\\
&= \quad a _A^\dag \circ a _A ^{\phantom{\dag}}\quad + \quad \id{A} \otimes \id{F(A)}.
\qedhere
\end{align*}

\end{proof}

\subsection{Physical interpretation and path-counting\label{pathcounting}}

For each object $A$ in \cat{C}, $F(A)$ is the generalised Fock space over $A$. We interpret the state \mbox{$e_A ^\dag : I \to F(A)$} as containing no `particles'; we think of it as the `vacuum state'. The morphism $\epsilon _A ^\dag : A \to F(A)$ injects the `single-particle space' into $F(A)$; if a state of $F(A)$ factors through this arrow, then we say that it is a `single-particle state'. We make this interpretation not only because it mirrors the particle interpretation for Fock space in quantum field theory; in fact, we will argue that it fits naturally with the extra structure that we have available.

The multiplication $d_A ^\dag : F(A) \otimes F(A) \to F(A)$ is interpreted as combining the states of two identical physical systems, to produce a state of a single physical system. Applying the adjoint, we obtain $d_A: F(A) \to F(A) \otimes F(A)$, which we interpret as constructing the superposition of every possible way of splitting the state of the system between two systems; or alternatively, of splitting the particle content of the system between two systems. The identity $d_A \circ \epsilon_A ^\dag = \epsilon_A^\dag \otimes e_A^\dag + e_A^\dag \otimes \epsilon_A^\dag$ now becomes transparent: splitting a single-particle state between two systems can be done in one of two ways, placing the single particle in either the first or second system, the other system being left in the vacuum state.

An even more concrete physical analogy is to think of $d_A : F(A) \to F(A) \otimes F(A)$ as a \emph{beamsplitter}, a half-silvered mirror which splits an obliquely incident stream of photons into two separate outgoing beams, one which passes through the mirror and one which reflects from it. In general, these outgoing beams will be entangled with one another, in state which is a superposition of all of the possible ways that the photon content of the incident beam could be distributed between the two outgoing beams.

We interpret the counit morphism $e_A : F(A) \to I$ as a measurement of the vacuum state. Essentially, it \emph{asserts} that there are no particles present, and in doing so contributes a scalar factor representing the amplitude that this is the case.

We will see that this `particle interpretation' of the components of the categorical structure gives rise to a useful algorithm for evaluating a large class of morphism compositions. The following diagrams must hold in our framework, the first by additivity of the comultiplication and the second by functoriality of $F$:
\secretlabel{superposition}
\secretlabel{following}
\begin{equation*}
\psfrag{n}{\hspace{1pt}\raisebox{-0.5pt}{$n$}}
\psfrag{m}{\hspace{0.5pt}\raisebox{-0.5pt}{$m$}}
\psfrag{npm}{\hspace{3pt}$n\mb+\mb m$}
\begin{array}{ccc}
\begin{array}{c}
\includegraphics{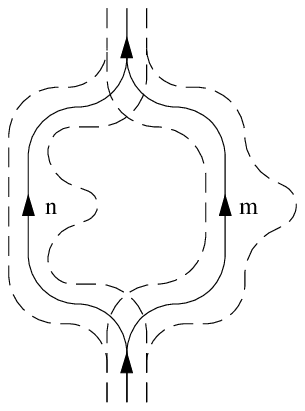}
\end{array}
=
\begin{array}{c}
\includegraphics{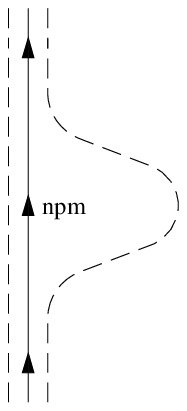}
\end{array}
&\quad\quad&
\begin{array}{c}
\includegraphics{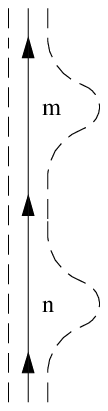}
\end{array}
=
\begin{array}{c}
\psfrag{npm}{\hspace{3pt}$m\circ n$}
\includegraphics{graphics/pathcounting2}
\end{array}
\\
(\ref{superposition})
&&
(\ref{following})
\end{array}
\end{equation*}
The symbols $n$ and $m$ are scalars $I \to I$ in $\cat{C}$, and $n+m$ is their sum as defined by the biproduct structure. The scalars are understood to be evaluated \emph{inside} the functor $F$, in the sense $F(n \cdot \id{A} )$. We can evaluate these diagrams by imagining a single, classical particle passing along them, in the direction of the arrows. For each possible route that the particle could take, we calculate a total amplitude for that route by composing the scalars $m$ or $n$ through which the particle passes. The total amplitude $w$ for the process is the sum of the amplitudes for each of the different routes, and the value for the diagram is $F(w \cdot \id{A})$. Diagram (\ref{superposition}) has two paths, contributing amplitudes $n$ and $m$ respectively; diagram (\ref{following}) has one path, contributing an amplitude $n \circ m$.

This can be considered an application of the \emph{path-integral} approach for quantum mechanics: the amplitude for a quantum process is a given by a sum of the amplitudes for each of the classical ways that the process could be completed. The novel aspect here is that the branching point at which the classical particle must make an exclusive choice is an explicit part of the theory, represented by the comultiplication morphisms $d_A$.

With this evaluation algorithm in mind, we would expect the following diagrams to evaluate to the same morphism, since in each case, there is only one path from each `in' leg to each `out' leg:
\begin{gather}
\centregraphics{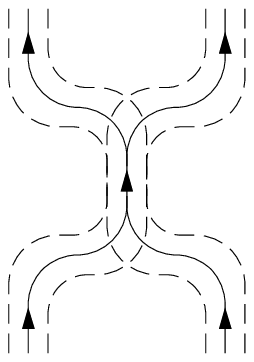}
=
\centregraphics{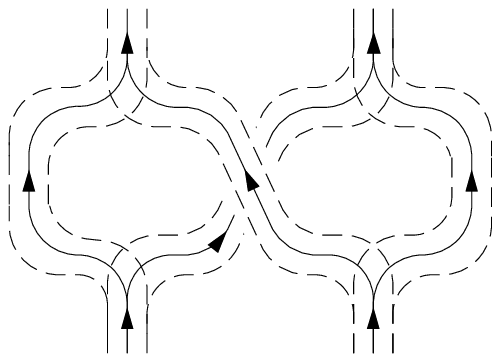}
\end{gather}
But this is precisely one of the axioms of a bialgebra, which we have already established in lemma \ref{bialgebralemma}.

We can now appreciate the difficulties that would be introduced here if our category had duals. In that case, we would be able to construct closed loops, and there could be an infinite number of paths by which to navigate a finite diagram in our category. However, lacking infinite scalars in general, we would not be able to make sense of this.

\subsection{Coherent states\label{coh}}

The harmonic oscillator adjunction $R\dashv Q$ implies particular isomorphisms of hom-sets, by the basic definition of an adjunction. Specifically, for all $(A,g,u) _\times$ in $\cat{C}_\times$ and all $B$ in $\cat{C}$ we have a natural isomorphism of sets:
\begin{equation}
\label{adjiso}
\quad H_{(A,g,u) _\times,B}:\Hom_\cat{C} ( A,B ) \simeq \Hom_{\cat{C}_\times}((A,g,u) _\times ,Q(B)).
\end{equation}
We recall the explicit definition for this isomorphism in terms of the unit and counit of the adjunction. If morphisms $f:A\to B$ in $\cat{C}$ and $f':(A,g,u) _\times \to Q(B)$ in $\cat{C}_\times$ are related by the isomorphism, then we must have the following:
\begin{align}
\label{iso1}
H_{(A,g,u) _\times,B}(f)&:=Q(f)\circ \eta^{} _{(A,g,u) _\times} = f'
\\
H_{(A,g,u) _\times,B}^{-1}(f')&:=\epsilon_B \circ R(f')=f.
\end{align}

Generally, objects $A$ in $\cat{C}$ do not have a natural cocommutative comonoid structure, so given only an $f:A\to B$ there will be no canonical choice of an $(A,g,u) _\times$ with which to implement equation (\ref{iso1}), producing a  morphism of comonoids in $\cat{C}_\times$. But there are exceptions: in particular, the monoidal unit object $I$ in $\cat{C}$ has a natural cocommutative comonoid structure \mbox{$I_\times := (I,\lambda_I^{\dag},\id{I})$}, which is the terminal object in $\cat{C}_\times$ as discussed in section \ref{tools}.

Making this natural choice of $I_\times$ for the comonoid on $I$, we can establish the following isomorphism for all objects $A$ in \cat{C}:
\begin{equation*}
H_{I_\times,A} : \Hom_{\cat{C}} (I,A) \simeq \Hom_{ \cat{C}_\times}(I_\times,Q(A)).
\end{equation*}
If we interpret elements of $\Hom_\cat{C}(I,A)$ as states of $A$, then this tells us that there is an isomorphism between states of $A$ in $\cat{C}$ and points of the comonoid $Q(A)$ in $\cat{C}_\times$. We can calculate the morphisms of comonoids to which these points correspond: given a state $\phi:I\to A$ in $\cat{C}$, we use equation (\ref{iso1}) to obtain a composite morphism
\begin{diagram}[midshaft,width=60pt]
\quad H_{I_\times,A}(\phi)=I_\times & \rTo^{ \eta^{} _{I_\times}} & (F(I),d_I,e_I) & \rTo^{ Q(\phi) } & (F(A),d_A,e_A) .
\end{diagram}
We interpret this as a generalised coherent state.
\begin{defn} The \emph{coherent state} $\Coh(\phi):I \to F(A)$ associated to the single-particle state $\phi:I \to A$ is given by\footnote{We note that this construction can still be made in the case of a model of linear logic on a category with biproducts, such as that described in \cite{differentialfiore}, where in place of $R \eta ^{} _{I _\times}$ we have \mbox{$\delta_0 : F(0) \simeq I \to FF(0) \simeq F(I)$}, where $\delta:F \dot{\to} FF$ is the comultiplication for the comonad.}
\begin{equation}
\quad \Coh(\phi) := F(\phi) \circ R\eta ^{} _{I_\times} .
\end{equation}
\end{defn}
\noindent We make this interpretation because states of this form have the properties that we expect from coherent states, as described in section \ref{conventionalqho}: they can be copied, and they are eigenstates of the lowering morphisms.

\begin{theorem} Coherent states are copied and deleted by the comultiplication and counit morphisms defined by the harmonic oscillator adjunction, as described by the following equations@
\begin{align}
d_A \circ \Coh(\phi) &= \Coh(\phi) \otimes \Coh(\phi)
\label{copycoh}
\\
e_A \circ \Coh(\phi) &= \id{I}.
\label{deletecoh}
\end{align}
\end{theorem}
\begin{proof} To prove equation ($\ref{copycoh}$), we use the fact that $F(\phi)$ and $R \eta _{I _\times}$ are morphisms of comonoids.
\begin{diagram}[midshaft,width=50pt,height=15pt]
I & \rTo ^{\hspace{-10pt} \Coh(\phi) \hspace{-10pt}} & & & F(A) & \rTo^{d_A} & F(A) \otimes F(A)
\\
\dEq & & & & \dEq & &
\\
I & \rTo ^{R \eta ^{} _{I _\times}} & F(I) & \rTo^{F( \phi) } & F(A) & &
\\
& & \dEq & & & & \dEq
\\
& & F(I) & \rTo^{d_I} & F(I) \otimes F(I) & \rTo^{F(\phi) \otimes F(\phi)} & F(A) \otimes F(A)
\\
\dEq & & & & \dEq & &
\\
I & \rTo^{ d_0 = \lambda_I ^\dag} & I \otimes I & \rTo^{ R \eta ^{} _{I _\times} \otimes R \eta ^{} _{I _\times}} & F(I) \otimes F(I) & &
\\
\dEq & & & & & & \dEq
\\
I & \rTo^{ \hspace{-20pt}\Coh(\phi) \otimes \Coh(\phi) \hspace{-20pt}} & & & & & F(A) \otimes F(A)
\end{diagram}
For equation (\ref{deletecoh}), we note that $e_A \circ \Coh(\phi) = e_A \circ F(\phi) \circ R \eta ^{} _{I _\times}$ is in the image of the hom-set $\Hom _{\cat{C} _\times} (I _\times,I_\times)$ under the functor $R : \cat{C} _\times \to \cat{C}$. But $I _\times$ is terminal, and functoriality of $R$ then implies $e_A \circ \Coh(\phi) = \id{I}$. 
\end{proof}

\begin{theorem}
\label{thmcohloweringeigenstate}
Coherent states are eigenstates of the lowering morphisms, satisfying the equation
\begin{equation}
\a{\psi} \circ \Coh(\phi)=(\psi^\dag\circ\phi) \cdot \Coh(\phi) .
\end{equation}
\end{theorem}
\begin{proof}
The theorem is proved by the following commuting diagram. In order, we employ that $F(\phi)$ is a morphism of comonoids, that $\eta ^{} _{I_\times}$ is a morphism of comonoids, the construction of the $H_{I_\times,A}$ isomorphism, and the properties of the symmetric monoidal structure.
\begin{diagram}[midshaft,loose,height=15pt,width=40pt]
I & \rTo^{\hspace{-20pt} \Coh(\phi) \hspace{-20pt}} & & & F(A) & \rTo ^{a_{\psi}} & & & & & F(A)
\\
\dEq & & & & \dEq & & & & & & \dEq
\\
I & \rTo^{\eta^{}_{I_\times}} &  F(I) & \rTo{\hspace{-10pt} F(\phi) \hspace{-10pt}} & F(A)
& \rTo^{d_A} & F(A) & \rTo ^{\id{F(A)} \mb\otimes\mb \epsilon_A} & F(A)
\mb\otimes\mb A & \rTo^{\id{A} \mb\otimes\mb \psi^\dag} & F(A) \\
&&\dEq&&&&\dEq&&&&\\
&&F(I)&\rTo^{d_I} & F(I) ^{\otimes 2} & \rTo ^{F(\phi) ^{\otimes 2}} & F(A)
\\
\dEq &&&&\dEq&&&&&&
\\
I&\rTo^{\lambda_I^{-1}}&I \mb\otimes\mb I
&\rTo^{ R\eta_{I_\times}^{\otimes 2} \hspace{-0pt} } & F(I) ^{\otimes 2} &&&&&&\\
&&\dEq&&&&&&\dEq&&\\
&&I \mb\otimes\mb I&\rTo^{\hspace{-40pt} \left( F(\phi) \circ R\eta^{}_{I_\times} \right) \otimes \phi\hspace{-40pt}}&&&&&F(A) \mb\otimes\mb A \\
\dEq&&&&&&&&&&\dEq\\
I & \rTo{ \hspace{-40pt} (\psi^\dag \circ \phi) \cdot \Coh( \phi ) \hspace{-40pt} } &&&&&&&&&F(A)
\end{diagram}
We also give the same proof using the graphical calculus. In this form, it is much more readable, although it has exactly the same content as the previous symbolic proof. Here, the double horizontal parallel lines all represent the morphism \mbox{$R \eta ^{} _{I _\times}: I \to F(I)$}. The explicit use of the adjunction equation (\ref{Radjeqn}) is very clear here.
\begin{equation*}
\psfrag{p}{$\scriptstyle\phi$}
\psfrag{fd}{$\psi^\dag$}
\begin{array}{c}
\includegraphics{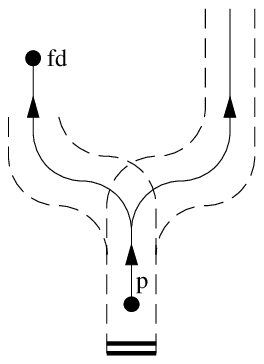}
\end{array}
=
\begin{array}{c}
\includegraphics{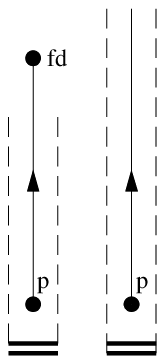}
\end{array}
\hspace{-6pt}
=
\hspace{-6pt}
\begin{array}{c}
\psfrag{bigp}{$\phi$}
\includegraphics{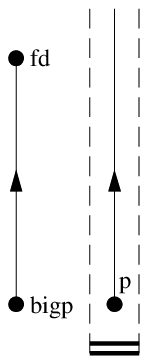}
\end{array}
\end{equation*}
\end{proof}

\section{Morphism exponentials}
\label{endomorphexp}
\subsection{Construction}

Let \cat{C} be a symmetric monoidal $\dag$-category with $\dag$-biproducts. We will see how a harmonic oscillator adjunction $\llangle Q, \eta, \epsilon \rrangle$ can be used to create exponentials of elements of commutative monoids in \cat{C}; that is, we will use commutative monoids $(A,g,u) _+$ to convert states $\phi : I \to A$ in \cat{C} into exponential states $\exp _{(A,g,u) _+}(\phi):I \to A$ in \cat{C}.

Since we will often find ourselves needing to work with with commutative monoids as well as cocommutative comonoids, we first introduce some new notation.

\begin{defn} Given a free cocommutative comonoid functor $Q:\cat{C} \to \cat{C} _\times$, we define the associated free commutative monoid functor to be $Q_+:\cat{C} \to \cat{C} _+$. It is therefore natural to refer to $Q$ itself as $Q _\times$. We also define the associated forgetful functors $R _\times : \cat{C} _\times \to \cat{C}$ and $R_+:\cat{C} _+ \to \cat{C}$. The adjunction $R_\times \dashv Q_\times$ with unit and counit $\eta ^\times$ and $\epsilon ^\times$ induces an adjunction $Q_+ \dashv R_+$ with unit and counit $\epsilon^+$ and $\eta^+$, where for all monoids $(A,g,u) _+$ and \cat{C}-objects $A$, $R_+ \eta^+ _{(A,g,u) _+} = (R_\times \eta^\times _{(A, g^\dag, u^\dag) _\times}) ^\dag$ and $\epsilon^+_A = (\epsilon_A ^\times) ^\dag$.
\end{defn}

\begin{defn} Given a commutative monoid $(A,g,u) _+$ in \cat{C}, and a state $\phi:I \to A$ in \cat{C}, then the \emph{exponential} $\exp _{(A,g,u) _+} (\phi):I \to A$ in \cat{C} is defined in the following way, where we give both the diagrammatic and graphical representations:
\begin{equation}
\begin{array}{c}
\begin{diagram}[midshaft,height=10pt,width=40pt]
I & \rTo^{\exp _{(A,g,u) _+}(\phi)} & & & & & A
\\
\dEq & & & & & & \dEq 
\\
I & \rTo^{R _\times \eta ^{\times} _{I _\times}} & F(I) & \rTo^{F(\phi)} & F(A) & \rTo^{R _+ \eta ^+ _{(A,g,u) _+}} & A
\end{diagram}
\end{array}
\hspace{30pt}
\begin{array}{c}
\psfrag{phi}{$\phi$}
\psfrag{eta2}{$R_+ \eta ^{+} _{(A,g,u) _+}$}
\psfrag{eta1}{$R_\times \eta ^{\times} _{I _\times}$}
\includegraphics{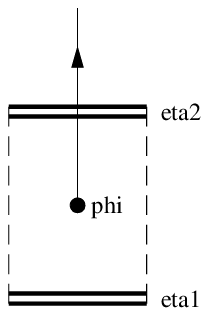}
\end{array}
\hspace{20pt}
\end{equation}
\end{defn}

\noindent The state $\phi$ is evaluated `inside a box' built from the categorical structure. Intuitively, we can think of $\exp _{(A,g,u) _+} (\phi)$ as being given by the infinite sum
\begin{align}
\exp _{(A,g,u) _+} (\phi) = \left(\frac{1}{0!} \cdot u \right) &+ \left( \frac{1}{1!}\cdot \phi \right) + \left(\frac{1}{2!} \cdot g \circ (\phi \otimes \phi) \right)
\nonumber
\\
& \quad+ \left(\frac{1}{3!} \cdot g  \circ (g \otimes \id{A}) \circ (\phi \otimes \phi \otimes \phi)\right) + \ldots,
\end{align}
as will in fact be the case in a suitable category of Hilbert spaces. The correspondence with the conventional notion of exponential is very clear.

This construction has several nice properties. First, we demonstrate that these exponentials behave well under composition, and that the unit for the monoid is given by the exponential of the zero morphism.
\begin{lemma} Exponentials of elements of commutative monoids compose additively; that is, for morphisms $\phi,\psi:I \to A$ and $(A,g,u) _+$,
\begin{equation}
g^\dag \circ \big( \exp(\phi) \otimes \exp(\psi) \big) = \exp(\phi + \psi)
\end{equation}
where the exponentials are defined with respect to the monoid $(A,g,u) _+$, and $\phi + \psi$ is defined by the biproduct structure in the underlying category.
\end{lemma}
\begin{proof} We can prove the lemma in a straightforward way using the graphical representation. We employ that $R _\times \eta ^\times _{I _\times}$ is a morphism of comonoids, that $R_+ \eta^+ _{(A,g,u) _+}$ is a morphism of monoids, and the additivity of the comultiplication for comonoids in the image of $Q$.
\begin{gather*}
\psfrag{phi}{$\phi$}
\psfrag{psi}{$\psi$}
\psfrag{pp}{$\phi+\psi$}
\psfrag{leta1}{\hspace{-15pt}$R_\times \eta ^{\times} _{I _\times}$}
\psfrag{leta2}{\hspace{-38pt}$R_+ \eta^+ _{(A,g,u) _+}$}
\psfrag{eta1}{\hspace{0pt}$R_\times \eta ^{\times} _{I _\times}$}
\psfrag{eta2}{$R_+ \eta ^{+} _{(A,g,u) _+}$}
\begin{array}{c}
\includegraphics{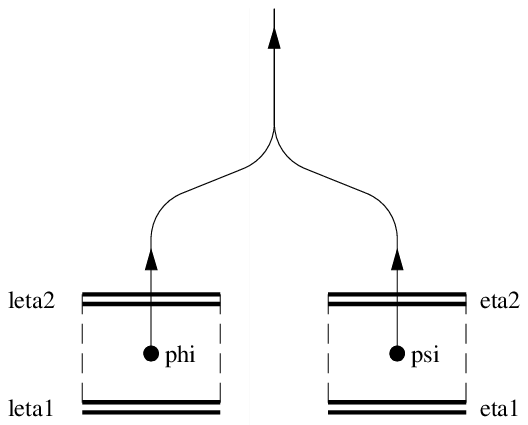}
\end{array}
\hspace{30pt}
=
\begin{array}{c}
\includegraphics{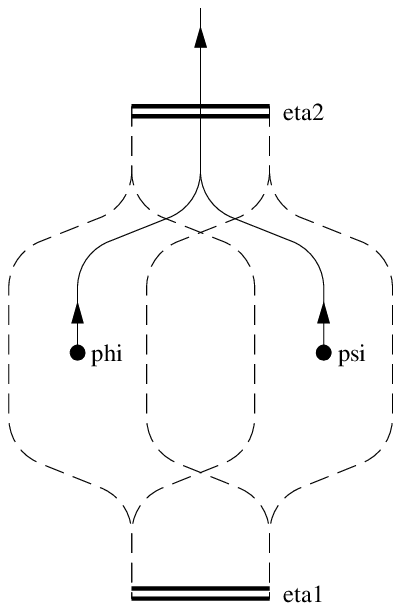}
\end{array}
\\
\psfrag{pp}{\hspace{3pt}$\phi+\psi$}
\psfrag{eta1}{$R_\times \eta ^{\times} _{I _\times}$}
\psfrag{eta2}{$R_+ \eta ^{+} _{(A,g,u) _+}$}
=
\begin{array}{c}
\includegraphics{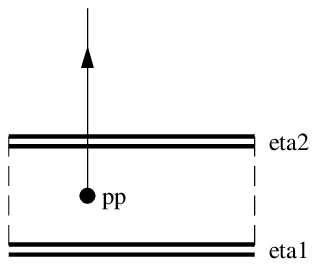}
\end{array}
\end{gather*}
\end{proof}

\begin{lemma} For all monoids $(A,g,u) _+$, we have
\begin{equation}
\exp _{(A,g,u) _+} (0_{I,A}) = u.
\end{equation}
\end{lemma}
\begin{proof} We demonstrate the lemma with the following diagram. We introduce the identity in the form $R k_0 ^{-1} \circ R k_0$, and employ the fact that $I _\times$ is terminal in $\cat{C} _\times$, the definition of $e_A ^\dag$ in terms of $F(0_{I,0})$, and the fact that $R _+ \eta ^+ _{(A,g,u) _+}$ is a morphism of monoids.
\begin{diagram}[midshaft,height=15pt]
I & \rTo^{R_\times \eta^{\times} _{I _\times}} & F(I) & \rTo^{F(0 _{I,A})} & & & & F(A) & \rTo ^{R_+ \eta^+ _{(A,g,u) _+}} & A
\\
& & \dEq & & & & & \dEq & &
\\
& & F(I) & \rTo^{F(0 _{I,0})} & F(0) & \rTo^{F( 0 _{0,A})} & & F(A) & &
\\
& & & \ruEq(2,2) & & \rdEq(2,2) & & & &
\\
& & F(0) & \rTo_{R k_0} & I & \rTo_{ R k_0 ^{-1}} & F(0) & & &
\\
\dEq & & & & \dEq & & & \dEq & &
\\
I & \rEq & & & I & \rTo^{e_A ^\dag} & & F(A) & &
\\
& & & & \dEq & & & & & \dEq
\\
& & & & I & \rTo^u & & & & A
\end{diagram}
\end{proof}

Finally, we demonstrate that the morphism exponential construction is natural in the following sense.
\begin{lemma}
\label{naturalexp}
For any pair of commutative monoids $(A,g,u)_+$ and $(B,h,v) _+$, and a morphism of monoids \mbox{$m:(A,g,u) _+ \to (B,h,v) _+$}, the following naturality condition holds for all $\phi: I \to A$:
\begin{equation}
\exp _{(B,h,v) _+} (m \circ \phi) = m \circ \exp_{(A,g,u) _+} (\phi).
\end{equation} 
\end{lemma} 
\begin{proof} Straightforward, from naturality of $\eta$.
\end{proof}

\noindent Intuitively, this implies that the exponential construction is only sensitive to the smallest submonoid containing the element to be exponentiated.

Given an element $\alpha_{N}$ of a noncommutative monoid $(N,k,s)_{\n+}$, where the subscript `$\n$' stands for noncommutative, we cannot apply this exponential construction directly, as the counit $\eta ^+$ associated to the harmonic oscillator adjunction is only defined for commutative monoids. However, inspired by lemma \ref{naturalexp}, if we have some commutative monoid $(A,g,u) _+$ and a morphism of monoids $m:(A,g,u) _+ \to (N,k,s) _{n+}$ such that for some $\alpha_A:I \to A$ we have $m \circ \alpha_A = \alpha_N$, then we may define
\begin{equation}
\label{expnoncomm}
\exp_{(M,k,s) _{n+}}(\alpha_M) := m \circ \exp_{(A,g,u) _+} (\alpha_A).
\end{equation}
We calculate the exponential using the commutative monoid, and embed the result into $N$ by using $m$.

\subsection{Endomorphism exponentials}

A common application of exponentials in functional analysis is to construct the exponential of an operator on a Hilbert space, defined using the familiar power series expansion. The analogue in our setting is to construct the exponential of an endomorphism $f:A \to A$.

To apply the generalised exponential construction which we have just developed, it seems that we would require a monoid which `knows' about arrow composition. Such a monoid is canonically present if our category has duals\footnote{This is no great surprise: a monoidal category has duals if, seen as a 2-category in a particular canonical way, each 1-morphism has a categorical adjoint. But each such adjunction then induces a monad in the familiar way, and this monad is precisely the monoid that we define here.}, as defined in section \ref{dualssection}.

\begin{defn}
In a monoidal category with duals, the \emph{endomorphism monoid} $\bigcirc_{\n+} ^A$ on an object $A$ is defined in the following way:
\begin{equation*}
\bigcirc _{\n+} ^A := \big( A \otimes A^*,\id{A} \otimes \theta _A \otimes \id{A} , \zeta_A \big) _{\n+}.
\end{equation*}
\end{defn}

\noindent Intuitively, elements of the monoid are names of endomorphisms on $A$, multiplication is endomorphism composition, and the unit for the monoid is the name of the identity morphism.

We can use this monoid to define endomorphism exponentials.
\begin{defn}
In a symmetric monoidal $\dag$-category with duals and $\dag$-biproducts, and a harmonic oscillator adjunction, the name of the exponential of an arbitrary endomorphism $f:A \to A$ is given by
\begin{equation}
\name{\exp{(f)}} := \exp_{\bigcirc ^A _{\n+}} (\name{f}),
\end{equation}
where the exponential over a noncommutative monoid is defined by equation (\ref{expnoncomm}).
\end{defn}

\noindent To construct the exponential of a particular operator, we would therefore need to find a commutative monoid, which embeds as a monoid into $\bigcirc _{\n+} ^A$, and through which the name of our operator factors. Although this is not straightforward, it intuitively seems likely that each endomorphism should generate a commutative monoid, and that this generated monoid should embed into the full endomorphism monoid. While it may not be straightforward in practice to construct such a commutative monoid, there does exist a canonical embedding into the endomorphism monoid.

\begin{lemma} In a monoidal category with duals, any monoid $(N,k,s) _{n+}$ has a canonical monic embedding $m_{(N,k,s) _{\n+}}:(N,k,s) _{\n+} \to \bigcirc_{\n+} ^N$ into the endomorphism monoid on $N$, given by
\begin{equation}
m_{(N,k,s)_{\n+}}:= (k \otimes \id{N^*}) \circ (\id{N} \otimes \zeta _N).
\end{equation}
\end{lemma}
\noindent Note that, in particular, commutative monoids will also have such an embedding.
\begin{proof} The embedding has the following graphical representation:
\begin{equation*}
\begin{diagram}[height=30pt]
N \otimes N ^*
\\
\uTo <{m _{(N,k,s)_{\n+}}}
\\
N
\end{diagram} =
\begin{array}{c}
\psfrag{h}{\raisebox{5pt}{\hspace{-10pt}${}$}}
\includegraphics{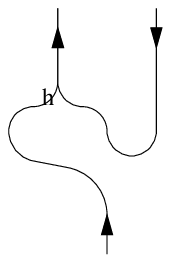}
\end{array}
\end{equation*}
Intuitively, each element of the monoid $(N,k,s) _{\n+}$ is taken to the name of the endomorphism of $A$ which performs multiplication by that element. It is simple to demonstrate that $m_{(N,k,s) _{\n+}}$ is monic, since it has a retraction: $(\id{N} \otimes s^*) \circ m_{(N,k,s) _{n+}} = \id{N}$, although note that this retraction will not be a morphism of monoids in general.

We first show that $m_{(N, k, s) _{n+}}$ preserves multiplication. We must demonstrate that following diagram commutes:
\begin{equation}
\begin{diagram}[midshaft,height=15pt]
N \otimes N & \rTo ^{ m_{(N,k,s) _{\n+}} \otimes m_{(N,k,s) _{\n+}}} & N \otimes N^* \otimes N \otimes N^* & \rTo ^{ \id{N} \otimes \theta _N \otimes \id{N^*}} & N \otimes N^*
\\
\dEq & & & & \dEq
\\
N \otimes N & \rTo ^{k} & N & \rTo ^{m_{(N,k,s) _{\n+}}} & N \otimes N^*
\end{diagram}
\end{equation}
Using the graphical representation, it can be seen that this equation holds, using the dual equations and associativity\footnote{Associativity of $(N,k,s) _+$ is clearly necessary for performing this embedding. For this reason, is does not seem useful to consider the commutative Jordan `algebra' generated by the monoid, as this will fail to be associative in general.} of the multiplication:
\begin{gather*}
\label{needtokillpi}
\begin{diagram}[height=25pt]
N \otimes N ^*
\\
\uTo > {\,\, \id{N} \otimes \theta_N \otimes \id{N}}
\\
N \otimes N^* \otimes N \otimes N^*
\\
\uTo > \,\,m_{(N,k,s) _{\n+}} ^{\otimes 2}
\\
N \otimes N
\end{diagram}
=
\psfrag{pi}{\hspace{-20pt}\raisebox{-2pt}{$\ts ^2_A {}^\dag \circ \ts ^2 _A$}}
\begin{array}{c}
\includegraphics{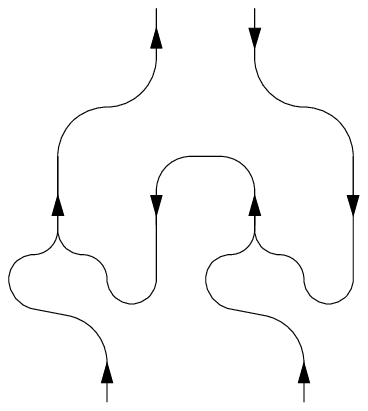}
\end{array}
=
\begin{array}{c}
\includegraphics{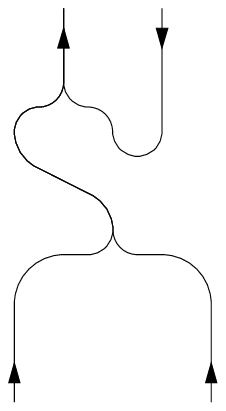}
\end{array}
\hspace{-10pt}
=
\begin{diagram}[height=25pt]
N \otimes N ^*
\\
\uTo > {\,\,m_{(N,k,s) _{\n+}}}
\\
N
\\
\uTo > \,\,k
\\
N \otimes N
\end{diagram}
\end{gather*}
Secondly, we show that $m_{(N,k,s) _{\n+}}$ preserves units, where we employ the unit law for the monoid $m_{(N,k,s) _{\n+}}$:
\begin{equation*}
\begin{diagram}[height=20pt]
N \otimes N^*
\\
\uTo<{m_{(N,k,s) _{\n+}}\,\,}
\\
N
\\
\uTo<{s\,\,}
\\
I
\end{diagram}
=
\begin{array}{c}
\includegraphics{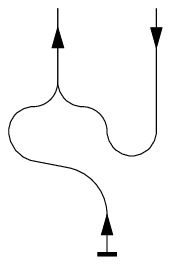}
\end{array}
=
\begin{array}{c}
\includegraphics{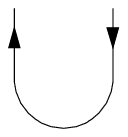}
\end{array}
=
\hspace{5pt}
\begin{diagram}[height=25pt]
N \otimes N^*
\\
\uTo<{\zeta _N}
\\
I
\end{diagram}
\end{equation*}
So $m_{(N,k,s) _{\n+}}$ is a well-defined morphism of monoids, for all monoids $(N,k,s) _{\n+}$.
\end{proof}

We now describe a quantum-mechanical application of this construction.
From the theory of the conventional quantum harmonic oscillator, we expect that for all objects $A$ and for all $\phi : I \to A$ in $\cat{C}$,
\begin{equation}
\label{suggestiveexpeqn}
\quad \Coh(\phi)=\exp(a^\dag_{\phi}) \circ e_A ^\dag,
\end{equation}
where $\exp(a_\phi ^\dag)$ is the endomorphism exponential of $a_\phi ^\dag$.

The name $\name{a _\phi ^\dag}$ of the raising morphism $a _\phi ^\dag : F(A) \to F(A)$ is an element of the endomorphism monoid $\bigcirc _{\n+} ^{F(A)}$. Noting that $\name{a _\phi ^\dag} = m_{Q_+\!(A)} \circ \epsilon ^\dag_A \circ \phi$, we apply equation (\ref{expnoncomm}) to obtain
\begin{align*}
\exp _{\bigcirc _{\n+} ^{F(A)}} \big(\name{a _\phi ^\dag} \big) &= \exp _{\bigcirc _{\n+} ^{F(A)}} \big( m_{Q_+\!(A)} \circ \epsilon ^\dag _A \circ \phi\big)
\\
&= m_{Q_+ (A)} \circ \exp_{Q_{+}\!(A)} \big( \epsilon_A ^\dag \circ \phi \big)
\end{align*}
This simplifies further, using the following lemma.
\begin{lemma} The coherent state is expressed as an exponential in the following way:
\label{cohasexp}
\begin{equation}
\Coh(\phi) = \exp _{Q_+ \! (A)} \big( \epsilon_A ^\dag \circ \phi \big).
\end{equation}
\end{lemma}
\begin{proof}
Writing this out with the graphical calculus, the result follows straightforwardly from the adjunction equation (\ref{Qadjeqn}).
\end{proof}

Writing the endomorphism exponential $\exp(a _\phi ^\dag)$ in terms of its name, we obtain
\begin{align*}
\exp(a _\phi ^\dag)
&\equiv
(\id{A} \otimes \theta_A) \circ \big( \name{\exp (a_\phi ^\dag)} \otimes \id{A}\big)
\\
&=
(\id{A} \otimes \theta_A) \circ \big(\exp _{\bigcirc _{\n+} ^{F(A)}} \big( \name{a _\phi ^\dag} \big) \otimes \id{A} \big)
\\
&= (\id{A} \otimes \theta_A) \circ \big(\big( m_{Q_+ \! (A)} \circ \Coh(\phi)\big) \otimes \id{A} \big).
\end{align*}
Writing graphically and simplifying, we obtain
\begin{equation}
\label{graphicalexpadagger}
\begin{diagram}[height=40pt]
F(A)
\\
\uTo>{\exp(a _\phi ^\dag) \textrm{}}
\\
F(A)
\end{diagram}
\hspace{8pt}=
\begin{array}{c}
\psfrag{p}{\raisebox{6pt} {\hspace{-3pt}$ \scriptstyle\phi$}}
\psfrag{eta}{\hspace{-1.5pt}$R_\times \eta ^{\times} _{I _\times}$}
\includegraphics{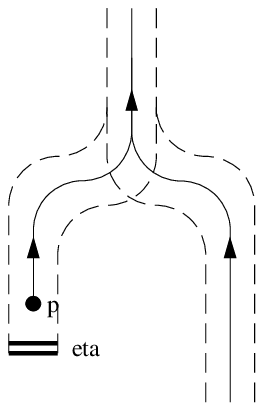}
\end{array}
\end{equation}

\begin{theorem}
The exponential of the raising morphism applied to the zero-particle state produces the coherent state:
\begin{equation*}
\quad \Coh(\phi)= \exp(a^\dag_{\phi}) \circ e_A ^\dag.
\eqno{(\ref{suggestiveexpeqn})}
\end{equation*}
\end{theorem}
\begin{proof}
The proof is immediate from the graphical representation (\ref{graphicalexpadagger}) and the unit law for the monoid $Q_+ (A)$.
\end{proof}

\section{The conventional quantum harmonic oscillator as a special case
\label{conventional}}

\subsection{Constructing the base category \label{constructinghilb}}

In this section, we will construct a categorical harmonic oscillator from a category of Hilbert spaces. The structures that we obtain will be  those  used to study the harmonic oscillator in conventional quantum mechanics. However, constructing the category in which we will work will not be straightforward. We will unavoidably need to work with infinite-dimensional separable\footnote{A separable Hilbert space is one of finite or countably-infinite dimension.} Hilbert spaces, and with unbounded linear maps, but the category of separable Hilbert spaces and all unbounded linear maps is not well-defined, as two unbounded maps $f:A \to B$ and $g:B \to C$ may not have a well-defined composite $g \circ f$. We overcome this problem by working instead with inner product spaces which are not necessarily complete; this will allow us to include  a restricted class of unbounded maps in our category. A second problem is that the trace operation for an infinite-dimensional Hilbert space is unbounded, and that these do not fit into our chosen class of admissible unbounded maps. As a consequence of this, our category will not be compact-closed.

Our category is defined as follows.
\begin{defn}
The category \cat{Inner} has objects given by countable-dimensional complex inner-product spaces. Morphisms are given by everywhere-defined linear maps $f:A \to B$, such that $f^\dag : B \to A$ is also everywhere-defined.
\end{defn}
\noindent Note that this category will contain both complete and non-complete inner-product spaces.
\begin{lemma}
The category \cat{Inner}  is a symmetric monoidal \mbox{$\dag$-category} with $\dag$-biproducts.
\end{lemma}
\begin{proof}
The symmetric monoidal structure is given by the tensor product of inner-product spaces, the monoidal unit object being given by the one-dimensional space $\mathbb{C}$. The $\dag$-structure is given by the adjoint, which by construction is well-defined. The $\dag$-biproduct structure is given by the direct sum of inner-product spaces, along with the canonical normalised and orthogonal projections.
\end{proof}
\noindent It may seem that we have abandoned Hilbert spaces, but we have done so only in a very mild way: every inner-product space $A$ has a canonical embedding into its completion $\bar{A}$, which is a Hilbert space. Moreover, all bounded linear operators $A \to B$ can be extended to bounded linear operators $\bar{A} \to B$, and further to bounded linear operators $\bar{A} \to \bar{B}$. In fact, this new definition makes no difference at all in the finite-dimensional case, since every finite-dimensional inner product space is complete, and therefore a Hilbert space; as a result, \cat{FdHilb} is a full subcategory of \cat{Inner}. However, in the infinite-dimensional case it makes a big difference: our unbounded operators can now be everywhere-defined, which ensures that they will always compose well with other morphisms.

Unfortunately, we will not be able to find duals in this category. The trace on an infinite-dimensional inner-product space $A$, considered as a linear map \mbox{$\mathrm{tr}_A:A^* \otimes A \to \mathbb{C}$} which takes the name of an operator to its trace, is not bounded, as the trace of a normalised operator name can be arbitrarily large. This is not a problem; the set of operator names having finite trace (the trace-class operators) do form an inner-product space. But the adjoint to this map has empty domain: $\mathrm{tr} _A ^\dag (1)=\name{\id{A}}$, but $\mathrm{tr}_A (\name{\id{A}})$ is infinite, and so $\mathrm{tr} _A ^\dag (1)$ is not an element of the space. This causes a problem for the construction of duals, for which we require a trace operator $\theta _A: A^* \otimes A \to \mathbb{C}$ and its transpose, the name of the identity operator $\name{\id{A}} \equiv \zeta _A: \mathbb{C} \to A \otimes A^*$, which we have shown cannot exist. To summarise: although the trace is a well-defined linear operator, it must be left out of our category since we lack an infinite scalar to represent the trace of the identity. Duals do not play a big part in our construction, but this does mean that we will be unable to implement the endomorphism exponential construction.

\subsection{Constructing the adjunction}

Although we are primarily interested in the category \cat{Inner} at this point, the construction of the adjunction that we are going to make can be carried out in a wide class of categories.

We work with a symmetric monoidal $\dag$-category \cat{C}, with countably-infinite $\dag$-biproducts, sufficient countably-infinite sums, and symmetric $\dag$-subspaces (which we define below.) Note that, although we require the existence of a countably-infinite extension to the $\dag$-biproduct structure, we do \emph{not} require the existence of countably-infinite diagonals and codiagonals. Rather, we require for each countable set $A_i$ of objects in \cat{C} the existence of a $\dag$-biproduct object $\bigoplus_i A_i$, along with canonical projection morphisms $\pi_j:\bigoplus_i A_i \to A_j$ for each element of the biproduct. 

\begin{defn}
The \emph{n-fold tensor product functor} $T_n : \cat{C} \to \cat{C}$ is defined on all objects $A$ and all morphisms $f$ as
\begin{align*}
T_n(A) &:= \bigotimesn{n} A =: A ^{\otimes n}
\\
T_n(f) &:= \bigotimesn{n} f =: f ^{\otimes n}
\end{align*}
$A ^{\otimes 0}$ is defined to be $I$, the tensor unit, and $f^{\otimes 0}$ is $\id{I}$.
\end{defn}

\begin{defn} Given a tensor product object $A ^{\otimes n}$, we write $(\swap_A ^\otimes) _{i,j} : A ^{\otimes n} \to A ^{\otimes n}$ for the symmetry isomorphism which exchanges the $i$th and $j$th factors of the tensor product.
\end{defn}

\begin{defn} We define the \emph{$n$-fold symmetric $\dag$-subspace} $A^{\otimes_s n}$ of an object $A$ as the following $\dag$-coequaliser,
\begin{equation*}
\begin{diagram}[midshaft,width=45pt,labelstyle=\scriptstyle]
A^{\otimes n} & \pile{\rTo^{\id{A ^{\otimes n}}}
\\
\rTo^{(\swap ^\otimes _A)_{1,2}}
\\
\vdots
\\
\rTo^{(\swap ^\otimes _A) _{(n-1),n}}} & A^{\otimes n} & \rTo^{s_A ^n} & A^{\otimes_s n} 
\end{diagram}
\end{equation*}
We are taking the coequaliser of all permutations of the tensor factors of $A ^{\otimes n}$, or at least of a generating set for the permutation group. $A^{\otimes_s 0}$ is defined to be $I$, the tensor unit.
\end{defn}

\noindent The $\dag$-coequaliser property implies that for all $A$ and $n$, we have $s_A ^n \circ s_A ^n {} ^\dag = \id{A ^{\otimes_s n}}$, as discussed in section \ref{daggerstructure}.

\begin{lemma} \label{naturalsymmproj} The projections onto the symmetric subspace $s _A ^n {}^\dag \circ s ^n _A : A^{\otimes n} \to A^{\otimes n}$ are natural.
\end{lemma}
\begin{proof} This follows from the fact that, for all arrows $f:A \to B$, $f ^{\otimes n} \circ \ts ^n_A {}^\dag \circ \ts ^n _A$ is a cone for the equaliser diagram defining $\ts _B ^n {}^\dag$.
\end{proof}

\begin{defn} The \emph{n-fold symmetric tensor product functor} $S_n:\cat{C} \to \cat{C}$ is defined in the following way, for all objects $A$ and all morphisms $f:A \to B$ in \cat{C}:
\begin{align*}
S_n(A) & := A ^{\otimes _s n}
\\
S_n(f) & := s ^n _B \circ f ^{\otimes n} \circ s ^n _A {}^\dag
\end{align*}
\end{defn}
\noindent Functoriality of composition follows from lemma \ref{naturalsymmproj}.

\begin{defn}
The \emph{symmetric subspace natural projection} $s^n: T_n \dot{\to} S_n$ is defined at each stage $A$ as $(s ^n) _A := s ^n _A$. It can be shown to be natural using lemma \ref{naturalsymmproj}.
\end{defn}

\begin{defn} The \emph{commutative ladder functor} $L _\cat{C} : \cat{C} \to \cat{C}$ on a symmetric monoidal $\dag$-category \cat{C} with $\dag$-biproducts and symmetric $\dag$-subspaces is defined as
\begin{align}
L := \bigoplus_{n=0} ^{\infty} S _n
\end{align}
where the $\dag$-biproduct is understood as being applied componentwise for each stage, so for all objects $A$ in \cat{C}, $G(A)= \bigoplus_{n=0} ^\infty S_n(A) = S_0 (A) \oplus S_1(A) \oplus \cdots$..
\end{defn}
\begin{defn}
The ladder functor $L$ has natural projections onto its components, which are derived from the $\dag$-biproduct structure. The projectors form the \emph{ladder projection natural transformations}, for all natural numbers $n$:
\begin{equation}
p ^n : L \dot{\to} S_n.
\end{equation}
Note that this is not an exponential notation; $n$ is just an index.
\end{defn}
\begin{defn} 
Any cocommutative comonoid $(A,g,u)_\times$ in $\cat{C}_\times$ has a series of \emph{$n$-fold comultiplication maps} $g^{n-1}:A\to A^{\otimes n}$ for all $n\geq 0$. The exponential notation is only suggestive; for the comultiplication $g:A\to A\otimes A$, of course, we cannot construct $g\circ g$. However, we \emph{can} construct $(g\otimes \id{A})\circ g:A\to A\otimes A\otimes A$. The associativity axiom means that this morphism is the same as $(\id{A}\otimes g)\circ g$, so there is no need to make a choice here. For $n\geq2$, we define
\begin{equation}
\quad g^{n-1}:= \left( g \otimes \id{A}^{(n-2)} \right) \circ \cdots \circ (g\otimes\id{A})\circ g\quad .
\end{equation}
For the special case of $n=0$, we define $g^{-1}:=u:A\to I$; and for $n=1$, we define $g^0:=\id{A}$. Also notice that $g^1\equiv g$, which gives the notation some intuitive consistency. If the notion of an $n$-fold comultiplication seems mysterious, that is only because we are not used to working with comonoids: it is no different to using the associative multiplication of a monoid to perform $n$-fold multiplication.
\end{defn}

\begin{defn}
\label{ladderstyle}
An adjunction $\langle R, P, \kappa, \lambda \rangle$, where $R: \cat{C} _\times \to \cat{C}$ is the forgetful functor, is \emph{ladder-style} if the right adjoint $P:\cat{C} \to \cat{C} _\times$, the unit and the counit are of the following form, for all objects $A$ and $B$ and all morphisms $f:A \to B$ in \cat{C}:
\begin{align*}
P(A) &= \big( L _\cat{C} (A), b_A, c_A \big) _\times
\\
P(f) &= L _\cat{C} (f)
\\
\\
b_A &= \sum_{m,n=0} ^\infty B_{m,n} \cdot \big( p ^m _A {}^\dag \otimes p_A ^n {}^\dag \big) \circ \big( s_A ^m \otimes s_A ^n \big) \circ \big( s^{m+n} _A \big) ^\dag \circ p^{m+n} _A
\\
c_A &= C \cdot p^0 _A
\\
\\
\kappa _{(A, g, u) _\times} &= \sum _{n=0} ^\infty K_n \cdot p^n _A {}^\dag \circ s^n _A \circ g^{n-1}
\\
\lambda_A &= L \cdot p^1 _A
\end{align*}
The objects $B_{m,n}$, $C$, $K_n$ and $L$ are all scalars in $\Hom _\cat{C} (I, I)$, and are functions of the natural numbers $m$ and $n$ where appropriate.
\end{defn}

\noindent These scalars cannot be freely chosen, as the counit, coassociativity and cocommutative equations, the adjunction equations, and the requirement that $\kappa _{ (A, g, u) _\times}$ is a morphism of comonoids will all give constraints on their values. There is also the need to ensure that all morphisms actually exist in the category $\cat{C}$, which will not necessarily be the case for all choices of scalars, given that all countable sums will not necessarily be defined. We will not enter into a complete analysis of the possible values that the scalars can take, but it is worth mentioning that setting all scalars equal to $\id{I}$ does satisfy all of the constraints, and so will provide an adjunction between \cat{C} and $\cat{C} _\times$ if all the required sums exist. This gives rise to the `canonical' free cocommutative comonoid construction.

\begin{lemma} If a ladder-style adjunction $\langle R, P, \kappa, \lambda \rangle$ is a harmonic oscillator adjunction, then the extra constraints imposed on the scalar coefficients, in addition to those arising from the condition that the adjunction be well-defined, are as follows:
\begin{enumerate}
\item Both $C$ and $L$ are unitary.
\item $\displaystyle B ^{\phantom{\dag}}_{n,m} \circ B_{n,m} ^\dag = \frac {(n+m)!} {n!\, m!}$.
\end{enumerate}
\end{lemma}

\noindent We note that $(n+m)! / n! \,m!$ is always a natural number, and so exists as a scalar in our category as a repeated sum of the scalar identity $\id{I}$.

\begin{lemma} If a ladder-style adjunction $\langle R, P, \kappa, \lambda \rangle$ is a harmonic oscillator adjunction, then the following equations must hold for the coefficients $K_n$\footnote{There is a fourth equation arising from the adjunction equation $P \lambda \circ \kappa P = \id{P}$, but this is not so easy to express in components.}:
\begin{enumerate}
\item $K_0 = C ^\dag$.
\item $K_1 = L ^\dag$.
\item $(n+1) K_{n+1} K_{n+1} ^\dag = K_n ^2 {}^\dag K_n ^2$.
\end{enumerate}
\end{lemma}

\noindent It is straightforward to see from the equations given that, in a category with inverses to all scalars, we must have
\begin{equation}
K_{n} ^{\phantom{\dag}} K_{n} ^\dag = \frac{1}{n!}.
\end{equation}
It is no surprise that these are the coefficients of the power series of the exponential function, given the role that the natural transformation $\kappa$ plays in defining the morphism exponentials.

We now investigate this structure in the category \cat{Inner}.

\begin{defn} In the $\dag$-category \cat{Inner}, we define the countably-infinite $\dag$-biproduct $\bigoplus_n A_n$ as the inner-product space consisting of sequences $a = (a_1,a_2,\ldots)$ of vectors, with $a_n \in A_n$, such that $\sum _{n=1} ^\infty c^n\, (a_n,a_n)_{A_n}$ is finite for all complex numbers $c$, and $(-,-)_{A_n}$ is the inner product of $A_n$. The inner product on $\bigoplus_n A_n$ is given by $(a,b) = \sum _{n=1} ^\infty (a_n,b_n) _{A_n}$, which will always be finite (just choose $c=1$.)
\end{defn}

We choose the rather strong requirement that  $\sum _{n=1} ^\infty c^n\, (a_n,a_n)_{A_n}$ be finite for all complex numbers $c$ to allow us to define a large class of linear maps on $\bigoplus_n A_n$. If we used the `standard' definition, where we only require the sum to converge for $c=1$, then the only everywhere-defined linear maps with source object $\bigoplus _n A_n$ would be the bounded linear maps. This would be too restrictive for our purposes; we will see that unbounded linear maps form a crucial part of the structure.

\begin{defn}
The harmonic oscillator adjunction $\llangle Q, \eta, \epsilon \rrangle$ on \cat{Inner} is a ladder-style adjunction, defined by coefficients
\begin{align*}
B _{n,m} &= \sqrt{ \frac {(n+m)!} {n!\,m!} },
&
K_n &= \frac{1} {\sqrt{n!}},
\\
C &= 1,
&
L &= 1.
\end{align*}
\end{defn}
\noindent There exist other harmonic oscillator adjunctions, but we choose this one as it is most straightforward.

We conjecture that this harmonic oscillator adjunction is well-defined. The difficulty lies in establishing that the natural transformation $\eta : \id{\cat{Inner} _\times} \dot{\to} Q \circ R$ is defined at every stage; more work is needed here.

Assuming the model is valid, it is simple to demonstrate that we recover the conventional raising and lowering operators from $\llangle Q, \eta, \epsilon \rrangle$:
\begin{align*}
a ^\dag _\phi &= d_A ^\dag \circ \big( \epsilon_A ^\dag \otimes \id{F (A)} \big) \circ \big(\phi \otimes \id{F (A)} \big)
\\
&= \sum_{m,n=0} ^\infty \sqrt \frac {(m+n)!} {m! \, n!} \cdot \big( p ^m _A {}^\dag \otimes p_A ^n {}^\dag \big) \circ \big( s_A ^m \otimes s_A ^n \big) \circ \big( s^{m+n} _A \big) ^\dag \circ p^{m+n}
\\
&= \sum _{m,n=0} ^\infty \sqrt \frac {(m+n)!} {m!\,n!} \cdot p_A ^{m+n} {}^\dag \circ s_A ^{m+n} \circ \big( s_A ^m {}^\dag \otimes s_A ^n {}^\dag \big) \circ \big( p_A ^m \otimes p_A ^n \big) \circ \big( \big( p_A ^1 {}^\dag \circ \phi \big)\otimes \id{F (A)} \big)
\\
&= \sum _{n=0} ^\infty \sqrt {n+1} \cdot p_A ^{n+1} {}^\dag \circ s_A ^{n+1} \circ \big( \id{A} \otimes s_A ^n {}^\dag \big) \circ \big( \phi \otimes p_A ^n \big)
\end{align*}
We also obtain the conventional coherent states:
\begin{align*}
\Coh(\phi) &= F_\H (\phi) \circ R_\H \eta ^{} _{I _\times}
\\
&= \left[ \sum _{n=0} ^\infty p_A ^n {}^\dag \circ s_A ^n \circ \phi ^{\otimes n} \circ s_I ^n {}^\dag \circ p_I ^n \right] \circ \left[ \sum_{m=0} ^\infty \frac 1 {\sqrt{m!}} \cdot p^m _I {}^\dag \circ s^m _I \circ \big(\lambda_I ^\dag \big)^{m-1} \right]
\\
&= \sum_{n=0} ^\infty \frac 1 {\sqrt{m!}} \cdot p_A ^n {}^\dag \circ s_A ^n \circ \phi ^{\otimes n} \circ \big( \lambda _I ^\dag \big) ^{n-1}
\\
&= \left( \id{I}, \phi, \frac 1 {\sqrt{2!}} \, \phi \otimes_s \phi, \frac{1} {\sqrt{3!}} \,\phi \otimes_s \phi \otimes_s \phi, \ldots, \frac{1} {\sqrt{n!}} \phi ^{\otimes _s n}, \ldots \right)
\end{align*}

We now turn to the morphism exponentials. Given a commutative monoid $(A,g,u) _+$ and an element $\phi:I \to A$, we can construct the exponential of $\phi$:
\begin{align*}
\exp_{(A,g,u) _+} (\phi) &= R_+ \eta ^+ _{(A,g,u) _+} \circ F(\phi) \circ R _\times \eta ^{\times} _{I _\times}
\\
&= \big( R \eta _{(A,g ^\dag, u ^\dag) _\times} \big) ^\dag \circ F(\phi) \circ R \eta _{I _\times} ^{}
\\
&= \sum _{m,n,p=0} ^\infty \frac{1} {\sqrt{m! \, p!}} \cdot g ^{m-1} \circ s_{A} ^m {}^\dag \circ p_{A} ^m \circ p ^n _{A} {}^\dag \circ s_{A} ^n \circ \phi ^{\otimes n} \circ s_I ^n {}^\dag
\\
& \hspace{250pt} {} \circ p_I ^n \circ p_I ^p {}^\dag \circ s_I ^p \circ \big( \lambda_I ^\dag \big) ^{p-1}
\\
&= \sum_{m=0} ^\infty \frac{1} {m!} \cdot g ^{m-1} \circ s_{A} ^m {}^\dag \circ s_{A} ^m \circ \phi ^{\otimes m} \circ s_I ^m {}^\dag \circ s_I ^m \circ \big( \lambda _I ^\dag \big) ^{p-1}
\\
&= \sum _{m=0} ^\infty \frac{1} {m!} \cdot g ^{m-1} \circ \phi ^{\otimes m}. \qquad \textrm{(suppressing $\lambda _I ^\dag$ isomorphisms)}
\end{align*}
This is clearly a direct generalisation of the conventional notion of exponential. We do not have duals in the category, so we cannot demonstrate endomorphism exponentials; however, if one allows the duality morphisms to `formally' exist in the category, it is clear the usual operator exponential construction is obtained.

\section{Discussion}
\label{discussion}

In this paper, we presented a categorical description of the quantum harmonic oscillator in any symmetric monoidal $\dag$-category with finite $\dag$-biproducts, giving categorical definitions of the zero-particle state, single-particle state and raising and lowering morphisms, and demonstrating the canonical commutators. We described how coherent states could be defined as the points of certain comonoids in $\cat{C} _\times$, and proved that they can be copied and deleted, and that they are eigenstates of the lowering morphisms. A formalism was developed for constructing exponentials of elements of commutative monoids, and we demonstrated that these exponentials are additive in a familiar way, and that the exponential of the appropriate zero element gives the unit for the monoid.

We then added duals to our category and demonstrated that endomorphism exponentials, generalising operator exponentials in conventional Hilbert space theory, could be defined in general, and that they interacted well with the existing categorical structures, coherent states being defined by the endomorphism exponentials of the raising morphisms.

Finally, we explored these structures in the context of a suitably-defined category of infinite-dimensional inner-product spaces, demonstrating that the conventional mathematical structures are produced by the categorical construction in this case. We conjectured that the construction was well-defined, and described the remaining open problem.

However, there are many issues which are still unclear. Philosophically, perhaps the biggest problem with the existing framework for categorical quantum mechanics is the lack of any nontrivial categorical description of dynamics. For this reason, it is questionable whether the system under study in this paper deserves to be called the harmonic oscillator at all: without a description of dynamics, all that has really been defined is the state space, but many different systems have isomorphic state spaces.

Conventionally, a quantum system is described by a Hilbert space $A$ and a self-adjoint operator $H:A \to A$. This operator describes the energy of different states of the system, and generates time evolution via Schr\"odinger's first-order differential equation. Work has been done on differential structure in categories, as we discuss below, but an elegant categorical description of the other ingredients of the Schr\"odinger equation, such as Planck's constant $h$ and the imaginary unit $i$, is far from apparent. The author is hopeful that the endomorphism exponential construction presented here might somehow prove relevant to this problem, given Stone's theorem in the classical case, which describes how the operator exponential induces an isomorphism between self-adjoint operators and certain unitary one-parameter groups.

We also note that the operator exponential construction, by way of the unit natural transformation $\eta$, calculates infinite sums, which, given that each term of the sum can be given a graphical representation, are strongly reminiscent of the infinite sums of Feynman diagrams which are used to defined an interacting field theory perturbatively. It would be interesting to explore whether a correspondence can be made here, and the formalism described in this paper extended to give a categorical framework for perturbative quantum field theory.

However, lurking behind the theory of the endomorphism exponential is the simple fact that interesting models of it seem impossible to come by. We recall from section \ref{conventional} that if a category of Hilbert spaces contains infinite-dimensional objects, then it cannot be made compact closed, since the trace of the identity on an infinite-dimensional space is infinite. There is no hope of adding an infinite scalar to the category: it is demonstrated in \cite{completionsemirings} that a semiring can only be given countably-infinite sums iff it admits a partial order, compatible with the semiring multiplication and addition, but the complex numbers admit no such ordering.

It seems that there is a fundamental incompatibility between duals, the existence of a right adjoint to the forgetful functor $R : \cat{C} _\times \to \cat{C}$, and non-trivial (by which we mean non-partially-orderable) scalars. The reason is this: if there exists a right adjoint $Q$ to $R$, then it seems that objects in the image of $Q$ must be infinite-dimensional. Duals allow us to take the trace of the identities on these infinite objects, but that will involve an infinite sum of scalars, which cannot be defined if the scalars are non-trivial. Dropping any one of these three properties, models are easy to come by: the category \cat{FdHilb} lacks a right adjoint to $R$, the category \cat{Inner} lacks duals, and the category \cat{Rel} of sets and relations lacks non-trivial scalars, but they all have the other two of the three properties.

A completely different interpretation of the categorical structures described in this paper is as a model for the differential calculus of polynomials in one or more variables (see \cite{bcsdifftl}, and also \cite{differentialfiore} developed in parallel to this work.) In this viewpoint, the raising operator performs multiplication by a homogeneous polynomial of degree 1, and the lowering operator performs differentiation. It is meaningless to discuss whether the mathematics presented here is more `honestly' a model of calculus or of the quantum harmonic oscillator, but there are certainly differences between the viewpoints: it is not conventional to consider multiplication and differentiation as being literally adjoint to each other, for example, as is natural for the quantum raising and lowering operators. The quantum point of view also suggests different avenues for generalisation: there are many other quantum systems, and it would be fascinating to explore the extent to which they too have natural categorical descriptions. Questions concerning the categorical descriptions of the dynamics also seem difficult to ask from the differential calculus point of view.

In fact, rather then being alternatives, the two viewpoints may have an interesting and nontrivial connection. As mentioned in section \ref{conventionalqho}, we expect the position and momentum operators --- that is, the multiplication and differentiation operators --- to arise as linear combinations of the raising and lowering operators, as $x _\phi = (a ^\dag _\phi + a ^{\phantom{\dag}} _\phi) / \sqrt 2$ and \mbox{$p _\phi = i \, (a_\phi ^\dag - a_\phi ^{ \phantom{ \dag}}) / \sqrt 2$} respectively, in a suitable choice of units. Note that these operators are both self-adjoint. They are of crucial importance: for example, the spectrum of $x _\phi$, where $\phi$ ranges over a basis for the single-particle space, represents physical space for the case of a harmonic oscillator which is free to oscillate in all directions.

It is also striking that the framework presented in this paper has so much in common with the model theory for linear logic \cite{differentialfiore,melliespanorama}. However, some aspects which are key to the description of the quantum harmonic oscillator --- such as the raising and lowering operators, and their commutation relations --- seem syntactically absent from linear logic, and are therefore quite mysterious from that point of view. In light of this, it would be interesting to explore whether linear logic might itself admit modification that would render it more suitable for performing deductions about quantum systems, such as the quantum harmonic oscillator considered here.

\section{Acknowledgements}
I would like to thank Bob Coecke, Marcelo Fiore, Robin Houston, Andrea Schalk and especially Chris Isham for useful discussions.

\bibliography{../../jamiebib}

\end{document}